\documentclass[11pt]{article}
\usepackage{amsmath, amssymb, graphicx, amsthm,ifthen,
  fancyhdr, color, cite, mathtools, float}
\usepackage[all]{xy}
\usepackage[boxruled]{algorithm2e}

\newcommand{\E}{\mathbb{E}}
\newcommand{\ind}{\mathbf{1}}
\renewcommand{\Pr}{\mathbb{P}}

\DeclareMathOperator{\tr}{tr}
\DeclareMathOperator{\argmin}{argmin}

\newtheorem{theorem}{Theorem}
\newtheorem{lemma}{Lemma}
\newtheorem{definition}{Definition}
\begin{document}
\title{Learning interactions through hierarchical group-lasso regularization}

\author{Michael Lim \and Trevor Hastie}

\newcommand{\fix}{\marginpar{FIX}}
\newcommand{\new}{\marginpar{NEW}}


\maketitle

\begin{abstract}
  We introduce a method for learning pairwise interactions in a manner that satisfies strong hierarchy: whenever an interaction is estimated to be nonzero, both its associated main effects are also included in the model. We motivate our approach by modeling pairwise interactions for categorical variables with arbitrary numbers of levels, and then show how we can accommodate continuous variables and mixtures thereof. Our approach allows us to dispense with explicitly applying constraints on the main effects and interactions for identifiability, which results in interpretable interaction models. We compare our method with existing approaches on both simulated and real data, including a genome wide association study, all using our R package \textsc{glinternet}.
\end{abstract}

\begin{section}{Introduction}

  Given an observed response and explanatory variables, we expect interactions to be present if the response cannot be explained by additive functions of the variables. The following definition makes this more precise.
  \begin{definition}
    \label{def:interactions}
    When a function $f(x,y)$ cannot be expressed as $g(x)+h(y)$ for some functions $g$ and $h$, we say that there is an interaction in $f$ between $x$ and $y$.
  \end{definition}
  
  Interactions of single nucleotide polymorphisms (SNPs) are thought to play a role in cancer \cite{Schwender:2008:Biostatistics} and other diseases. Modeling interactions has also served the recommender systems community well: latent factor models (matrix factorization) aim to capture user-item interactions that measure a user's affinity for a particular item, and are the state of the art in predictive power \cite{Koren:2009:KDD}. In lookalike-selection, a problem that is of interest in computational advertising, one looks for features that most separates a group of observations from its complement, and it is conceivable that interactions among the features can play an important role.

  There are many challenges, the first of which is a problem of scalability. Even with 10,000 variables, we are already looking at a $50\times10^6$-dimensional space of possible interaction pairs. Complicating the matter are spurious correlations amongst the variables, which makes learning even harder. Finally, in some applications, sample sizes are relatively small and the signal to noise ratio is low. For example, genome wide association studies (GWAS) can involve hundreds of thousands or millions of variables, but only several thousand observations. Since the number of interactions is on the order of the square of the number of variables, computational considerations quickly become an issue.

  Finding interactions is an example of the ``$p > n$'' problem where there are more features or variables than observations. A popular approach in supervised learning problems of this type is to use regularization, such as adding a squared $L_2$ penalty of the form $\|\beta\|_2^2$ or a $L_1$ penalty of the form $\|\beta\|_1$ to the coefficients. The latter type of penalty has been the focus of much research since its introduction in \cite{Tibshirani:1996:JRSS}, and is called the lasso. One of the reasons for the lasso's popularity is that it does variable selection: it sets some coefficients exactly to zero. There is a group analogue to the lasso, called the group-lasso \cite{Yuan:2006:JRSS}, that sets groups of variables to zero. The idea behind our method is to set up main effects and interactions (to be defined later) as a group of variables, and then we perform selection via the group-lasso. 

  Discovering interactions is an area of active research; see, for example, \cite{Bien:2013:AnnalsStat} and \cite{Chen:2011:CompBio}. In this paper, we introduce \textsc{glinternet}, a method for learning first-order interactions that can be applied to categorical variables with arbitrary numbers of levels, continuous variables, and combinations of the two. Our approach consists of two phases: a screening stage (for large problems) that gives a candidate set of main effects and interactions, followed by variable selection on the candidate set with the group-lasso. We introduce two screening procedures, the first of which is inspired by our observation that boosting with depth-2 trees naturally gives rise to an interaction selection process that enforces hierarchy: an interaction cannot be chosen until a split has been made on its associated main effect. The second method is an adaptive procedure that is based on the strong rules \cite{Tibshirani:2012:JRSS} for discarding predictors in lasso-type problems. We show in Section \ref{sec:methodology} how the group-lasso penalty naturally enforces strong hierarchy in the resulting solutions.

  We can now give an overview of our method:
  \begin{enumerate}
  \item If required, screen the variables to get a candidate set
    $\mathcal{C}$ of interactions and their associated main effects. Otherwise, take $\mathcal{C}$ to consist of all main effects and pairwise interactions. 
  \item Fit a group-lasso on $\mathcal{C}$ with a grid of values for the
    regularization parameter. Start with $\lambda=\lambda_{max}$ for which all
    estimates are zero. As we decrease $\lambda$, we allow more terms
    to enter the model, and we stop once a user-specified number of
    interactions have been discovered. Alternatively, we can choose $\lambda$ using any model selection technique such as cross validation.
  \end{enumerate}
  
  \begin{subsection}{A simulated example}\label{sec:simulated_example}
    As a first example, we perform 100 simulations with 500 3-level categorical variables and 800 quantitative observations. There are 10 main effects and 10 interactions in the ground truth, and the noise level is chosen to give a signal to noise ratio of one. We run \textsc{glinternet} without any screening, and stop after ten interactions have been found. The average false discovery rate and standard errors are plotted as a function of the number of interactions found in Figure \ref{fig:introduction_example}.\\
    \begin{figure}[h]
      \begin{center}
        \includegraphics[width=0.5\linewidth]{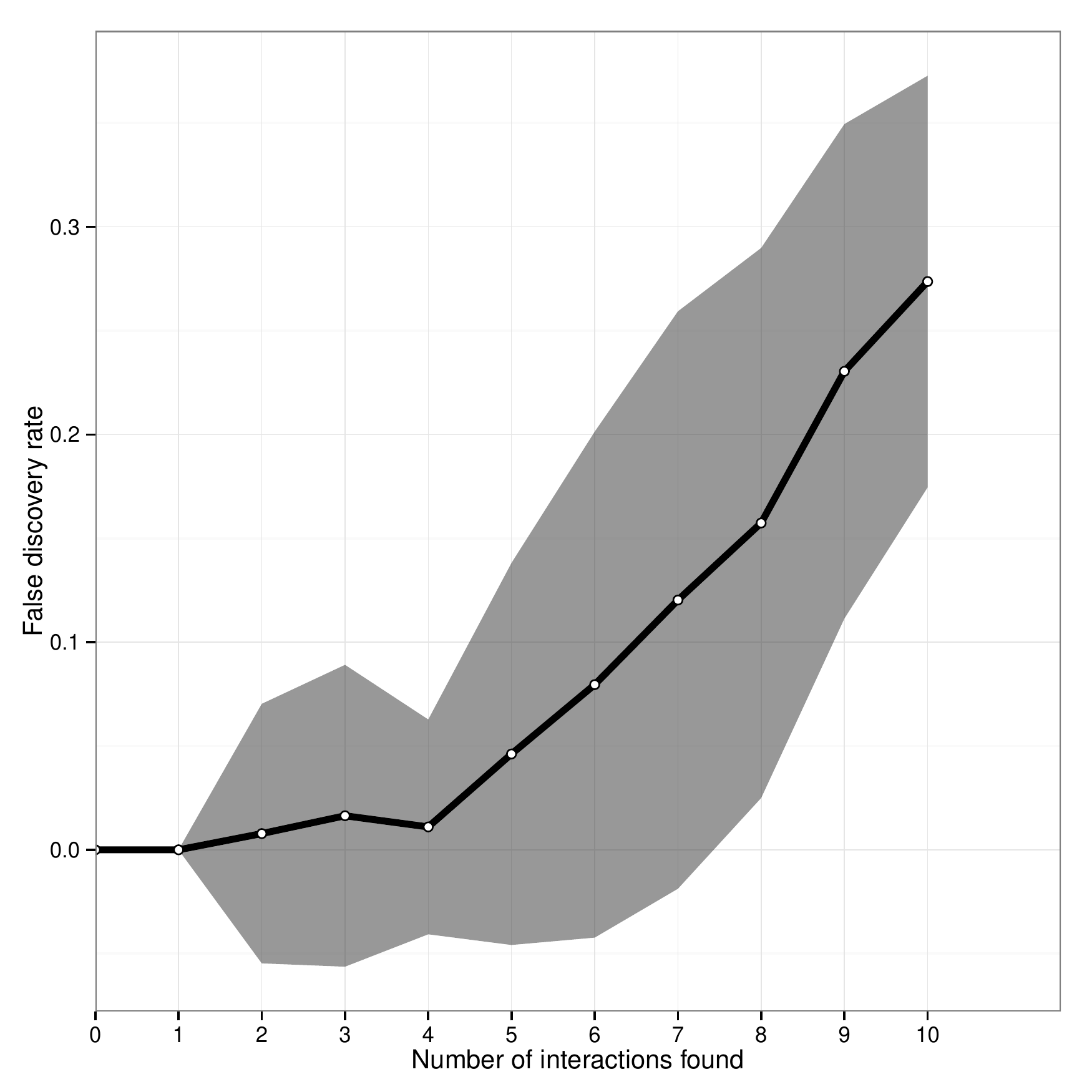}
      \end{center}
      \caption{False discovery rate vs number of discovered interactions}
      \label{fig:introduction_example}
    \end{figure}
  \end{subsection}

  \begin{subsection}{Organization of the paper}\label{sec:organization}
    The rest of the paper is organized as follows. Section
    \ref{sec:background_and_notation} introduces basic notions and
    notation. In section \ref{sec:methodology}, we introduce the group-lasso and how it fits into our framework for finding interactions. We also show how \textsc{glinternet} is equivalent to an overlapping grouped lasso. We discuss screening in Section \ref{sec:screening}, and give several examples with both synthetic and real datasets in Section \ref{sec:real_data} before going into algorithmic details in Section \ref{sec:algorithm_details}. We conclude with a discussion in Section \ref{sec:discussion}.
  \end{subsection}
\end{section}

\begin{section}{Background and notation}\label{sec:background_and_notation}
  
  We use the random variables $Y$ to denote the response, $F$ to denote a categorical feature, and $Z$ to denote a continuous feature. We use $L$ to denote the number of levels that $F$ can take. For simplicity of notation we will use the first $L$ positive integers to represent these $L$ levels, so that $F$ takes values in the set $\{i\in\mathbb{Z}: 1\leq i\leq L\}$. Each categorical variable has an associated random variable $X\in\mathbb{R}^L$ with a 1 that indicates which level $F$ takes, and 0 everywhere else.

 When there are $p$ categorical (or continuous) features, we will use subscripts to index them, i.e. $F_1,\ldots,F_p$. Boldface font will always be reserved for vectors or matrices that comprise of realizations of these random variables. For example, $\mathbf{Y}$ is the $n$-vector of observations of the random variable $Y$, $\mathbf{F}$ is the $n$-vector of realizations of the random variable $F$, and $\mathbf{Z}$ is the $n$-vector of realizations of the random variable $Z$. Similarly, $\mathbf{X}$ is a $n\times L$ indicator matrix whose $i$-th row consists of a 1 in the $\mathbf{F}_i$-th column and 0 everywhere else. We use a $n\times(L_i\cdot L_j)$ indicator matrix $\mathbf{X}_{i:j}$ to represent the interaction $F_i:F_j$. We will write
  \begin{eqnarray}
    \mathbf{X}_{i:j} = \mathbf{X}_i * \mathbf{X}_j,
  \end{eqnarray}
  where the first $L_j$ columns of $\mathbf{X}_{i:j}$ are obtained by taking the elementwise products between the first column of $\mathbf{X}_i$ and the columns of $\mathbf{X}_j$, and likewise for the other columns. For example,
  \begin{eqnarray}
    \begin{pmatrix}
      a & b\\
      c & d
    \end{pmatrix} *
    \begin{pmatrix}
      e & f\\
      g & h
    \end{pmatrix} =
    \begin{pmatrix}
      ae & af & be & bf\\
      cg & ch & dg & dh
    \end{pmatrix}.
  \end{eqnarray}

  \begin{subsection}{Definition of interaction for categorical variables}
    To see how Definition \ref{def:interactions} applies to this setting, let $\E(Y|F_1=i,F_2=j) = \mu_{ij}$, the conditional mean of $Y$ given that $F_1$ takes level $i$, and $F_2$ takes level $j$. There are 4 possible cases:
    \begin{enumerate}
    \item $\mu_{ij}=\mu$ (no main effects, no interactions)
    \item $\mu_{ij}=\mu+\theta_1^i$ (one main effect $F_1$)
    \item $\mu_{ij} = \mu+\theta_1^i+\theta_2^j$ (two main effects)
    \item $\mu_{ij} = \mu+\theta_1^i+\theta_2^j+\theta_{1:2}^{ij}$ (main effects and interaction)
    \end{enumerate}
  
    Note that all but the first case is overparametrized, and the usual procedure is to impose sum constraints on the main effects and interactions:
    \begin{eqnarray}
      \label{eq:main_effect_constraints}
      \sum_{i=1}^{L_1}\theta_1^i = 0, \quad \sum_{j=1}^{L_2}\theta_2^j = 0
    \end{eqnarray}
    and
    \begin{eqnarray}
      \label{eq:interaction_constraints}
      \sum_{i=1}^{L_1}\theta_{1:2}^{ij} = 0 \text{ for fixed $j$}, \quad \sum_{j=1}^{L_2}\theta_{1:2}^{ij} = 0 \text{ for fixed $i$}.
    \end{eqnarray}
    
    In what follows, $\theta_i$, $i=1,\cdots, p$, will represent the main effect coefficients, and $\theta_{i:j}$ will denote the interaction coefficients. We will use the terms ``main effect coefficients'' and ``main effects'' interchangeably, and likewise for interactions.
  \end{subsection}

  \begin{subsection}{Weak and strong hierarchy}
    An interaction model is said to obey strong hierarchy if an interaction can be present only if both of its main effects are present. Weak hierarchy is obeyed as long as either of its main effects are present. Since main effects as defined above can be viewed as deviations from the global mean, and interactions are deviations from the main effects, it rarely make sense to have interactions without main effects. This leads us to prefer interaction models that are hierarchical. We will see in Section \ref{sec:methodology} that \textsc{glinternet} produces estimates that obey strong hierarchy.
  \end{subsection}

  \begin{subsection}{First order interaction model}\label{sec:interaction_model}
    Our model for a quantitative response $Y$ is given by
    \begin{eqnarray}\label{eq:model_quantitative}
      Y = \mu + \sum_{i=1}^pX_i\theta_i + \sum_{i<j}X_{i:j}\theta_{i:j} + \epsilon
    \end{eqnarray}
    where $\epsilon\sim N(0,\sigma^2)$. For binary responses, we have 
    \begin{eqnarray}\label{eq:model_binary}
      logit(\Pr(Y=1|X)) = \mu+ \sum_{i=1}^pX_i\theta_i + \sum_{i<j}X_{i:j}\theta_{i:j}.
    \end{eqnarray}
    We fit these models by minimizing an appropriate choice of loss function $\mathcal{L}$. Because the models are still overparametrized, we impose the relevant constraints for the coefficients $\theta$ (see (\ref{eq:main_effect_constraints}) and (\ref{eq:interaction_constraints})). We can thus cast the problem of fitting a first-order interaction model as an optimization problem with constraints:
    \begin{eqnarray}
      \label{eq:basic_optimization_problem}
      \argmin_{\mu,\theta} \mathcal{L}(\mathbf{Y},\mathbf{X}_{i:i\leq p}, \mathbf{X}_{i:j}; \mu, \theta)
    \end{eqnarray}
    subject to the relevant constraints. $\mathcal{L}$ can be any loss function, typically squared error loss for the quantitative response model given by
    \begin{eqnarray}
      \mathcal{L}(\mathbf{Y}, \mathbf{X}_{i:i\leq p}, \mathbf{X}_{i:j}; \mu, \theta) = \frac{1}{2}\left\|\mathbf{Y} - \mu\cdot\ind - \sum_{i=1}^p\mathbf{X}_i\theta_i + \sum_{i<j}\mathbf{X}_{i:j}\theta_{i:j}\right\|_2^2,
    \end{eqnarray}
    and logistic loss for the binomial response model given by
    \begin{multline}
      \mathcal{L}(\mathbf{Y}, \mathbf{X}_{i:i\leq p}, \mathbf{X}_{i:j}; \mu, \theta) = -\left[\mathbf{Y}^T(\mu\cdot\ind+\sum_{i=1}^p\mathbf{X}_i\theta_i+\sum_{i<j}\mathbf{X}_{i:j}\theta_{i:j})\right.\\
      - \left.\ind^T\log\left(\ind+\exp(\mu\cdot\ind+\sum_{i=1}^p\mathbf{X}_i\theta_i+\sum_{i<j}\mathbf{X}_{i:j}\theta_{i:j})\right)\right],
    \end{multline}
    where the log and exp are taken component-wise. Notice that the optimization problem (\ref{eq:basic_optimization_problem}) does not impose any hierarchical constraints on the coefficients $\theta$, so that its solutions may obey neither weak nor strong hierarchy. This need not be a problem, although it is our contention that some kind of hierarchy is more natural.
  \end{subsection}

  \begin{subsection}{Group-lasso and overlapped group-lasso}\label{sec:group_lasso}
    Since \textsc{glinternet}'s workhorse is the group-lasso, we briefly introduce it here. We refer the reader to \cite{Yuan:2006:JRSS} for more technical details.

    The group-lasso can be thought of as a more general version of the well-known lasso. Suppose there are $p$ groups of variables (possibly of different sizes), and let the feature matrix for group $i$ be denoted by $\mathbf{X}_i$. Let $\mathbf{Y}$ denote the vector of observations. The group-lasso obtains the estimates $\hat{\beta}_j$ as the solution to
    \begin{eqnarray}\label{eq:group_lasso_objective}
      \displaystyle\argmin_{\mu,\beta}\frac{1}{2}\|\mathbf{Y}-\mu\cdot\ind-\sum_{j=1}^p\mathbf{X}_j\beta_j\|_2^2+\lambda\sum_{j=1}^p\gamma_j\|\beta_j\|_2.
    \end{eqnarray}
    Note that if each group consists of only one variable, this reduces to the lasso criterion. In our application, each indicator matrix $\mathbf{X}$ will represent a group. The group-lasso applied to the $\mathbf{X}$'s intuitively selects those variables that have strong overall contribution from all the levels toward explaining the response by setting some of the $\beta$'s to 0. If an estimate $\hat{\beta}_i$ is nonzero, then \textit{all} its components are usually nonzero.

    The parameter $\lambda$ controls the amount of regularization, with larger values implying more regularization. The $\gamma$'s allow each group to be penalized to different extents; we set them all equal to 1 (see Section \ref{sec:algorithm_details}). To solve (\ref{eq:group_lasso_objective}), we start with $\lambda$ large enough so that all estimates are zero. Decreasing $\lambda$ along a grid of values results in a path of solutions, from which an optimal $\lambda$ can be chosen by cross validation or some model selection procedure.

    The Karush-Kuhn-Tucker (KKT) optimality conditions for the group-lasso are simple to compute and check. For group $i$, they are
    \begin{eqnarray}
      \label{eq:kkt_conditions}
      \|\mathbf{X}_i^T(\mathbf{Y}-\hat{\mathbf{Y}})\|_2 < \gamma_i\lambda & \text{if}\quad\hat{\beta}_i=0\\
      \|\mathbf{X}_i^T(\mathbf{Y}-\hat{\mathbf{Y}})\|_2 = \gamma_i\lambda & \text{if}\quad\hat{\beta}_i\neq0.
    \end{eqnarray}
    The group-lasso is commonly fit with some form of gradient descent, and convergence can be confirmed by checking the KKT conditions. Details of the algorithm we use can be found in Section \ref{sec:algorithm_details}.

    The overlapped group-lasso is a variant of the group-lasso where the groups of variables are allowed to have overlaps, i.e. some variables can show up in more than one group. However, each time a variable shows up in a group, it gets a new coefficient. For example, if a variable is included in 3 groups, then it has 3 coefficients that need to be estimated.
  \end{subsection}
\end{section}

\begin{section}{Methodology and results}\label{sec:methodology}

  We want to fit the first order interaction model in a way that obeys strong hierarchy. We show in Section \ref{sec:overlapped_group_lasso} how this can be achieved by adding an overlapped group-lasso penalty to the objective in (\ref{eq:basic_optimization_problem}). We then show how this optimization problem can be conveniently solved via a group-lasso without overlaps.

  \begin{subsection}{Strong hierarchy through overlapped group-lasso}\label{sec:overlapped_group_lasso}
    Adding an overlapped group-lasso penalty to (\ref{eq:basic_optimization_problem}) is one way of obtaining solutions that satisfy the strong hierarchy property. The results that follow hold for both squared error and logistic loss, but we focus on the former for clarity.

    Consider the case where there are two categorical variables $F_1$ and $F_2$ with $L_1$ and $L_2$ levels respectively. Their indicator matrices are given by $\mathbf{X}_1$ and $\mathbf{X}_2$. We solve
    \begin{multline}
      \label{eq:overlapped_group_lasso}
      \argmin_{\mu,\alpha,\tilde{\alpha}} \frac{1}{2}\left\|\mathbf{Y}-\mu\cdot\ind-\mathbf{X}_1\alpha_1-\mathbf{X}_2\alpha_2-[\mathbf{X}_1\text{  }\mathbf{X}_2\text{  }\mathbf{X}_{1:2}]\left[
          \begin{array}{c}
            \tilde{\alpha}_1\\
            \tilde{\alpha}_2\\
            \alpha_{1:2}
          \end{array}
        \right]\right\|_2^2\\
      + \lambda\left(\|\alpha_1\|_2 + \|\alpha_2\|_2 + \sqrt{L_2\|\tilde{\alpha}_1\|_2^2 + L_1\|\tilde{\alpha}_2\|_2^2+\|\alpha_{1:2}\|_2^2}\right)
    \end{multline}
    subject to
    \begin{eqnarray}
      \label{eq:overlapped_main_effect_constraints}
      \sum_{i=1}^{L_1}\alpha_1^i = 0,\quad \sum_{j=1}^{L_2}\alpha_2^j = 0,\quad \sum_{i=1}^{L_1}\tilde{\alpha}_1^i = 0, \quad \sum_{j=1}^{L_2}\tilde{\alpha}_2^j = 0
    \end{eqnarray}
    and
    \begin{eqnarray}
      \label{eq:overlapped_interaction_constraints}
      \sum_{i=1}^{L_1}\alpha_{1:2}^{ij} = 0 \text{ for fixed $j$},\quad \sum_{j=1}^{L_2}\alpha_{1:2}^{ij} = 0 \text{ for fixed $i$}.
    \end{eqnarray}
    Notice that $\mathbf{X}_i$, $i=1,2$ each have two different coefficient vectors $\alpha_i$ and $\tilde{\alpha}_i$, resulting in an overlapped penalty. The $\sqrt{L_2\|\tilde{\alpha}_1\|_2^2 + L_1\|\tilde{\alpha}_2\|_2^2+\|\alpha_{1:2}\|_2^2}$ term results in estimates that satisfy strong hierarchy, because either $\hat{\tilde{\alpha}}_1=\hat{\tilde{\alpha}}_2=\hat{\alpha}_{1:2}=0$ or all are nonzero, i.e. interactions are always present with both main effects.

    The constants $L_1$ and $L_2$ are chosen to put $\tilde{\alpha}_1$, $\tilde{\alpha}_2$, and $\alpha_{1:2}$ on the same scale. To motivate this, note that we can write
    \begin{eqnarray}
      X_1\tilde{\alpha}_1 = X_{1:2}[\underbrace{\tilde{\alpha}_1,\ldots,\tilde{\alpha}_1}_{L_2\text{ copies}}]^T,
    \end{eqnarray}
    and similarly for $X_2\tilde{\alpha}_2$. We now have a representation for $\tilde{\alpha}_1$ and $\tilde{\alpha}_2$ with respect to the space defined by $X_{1:2}$, so that they are ``comparable'' to $\alpha_{1:2}$. We then have
    \begin{eqnarray}
      \|[\underbrace{\tilde{\alpha}_1,\ldots,\tilde{\alpha}_1}_{L_2\text{ copies}}]\|_2^2 = L_2\|\tilde{\alpha}_1\|_2^2
    \end{eqnarray}
    and likewise for $\tilde{\alpha}_2$. More details are given in Section \ref{sec:equivalence_with_group_lasso} below.

    The actual main effects and interactions can be recovered as
    \begin{eqnarray}
      \hat{\theta}_1 = \hat{\alpha}_1 + \hat{\tilde{\alpha}}_1\\
      \hat{\theta}_2 = \hat{\alpha}_2 + \hat{\tilde{\alpha}}_2\\
      \hat{\theta}_{1:2} = \hat{\alpha}_{1:2}.
    \end{eqnarray}
    Because of the strong hierarchy property mentioned above, we also have
    \begin{eqnarray}
      \hat{\theta}_{1:2} \neq 0 \Longrightarrow \hat{\theta}_1\neq0 \text{ and } \hat{\theta}_2\neq0.
    \end{eqnarray}

    The overlapped group-lasso with constraints is conceptually simple, but care must be taken in how we parametrize the constraints. This is especially so because we penalize the coefficients, and any representation of the problem that does not preserve symmetry will result in unequal penalization schemes for the coefficients. The problem becomes more tedious as the number of variables and levels grows. We now show how to solve the overlapped group-lasso problem by solving an equivalent \textit{unconstrained} group-lasso problem. This is advantageous because
    \begin{enumerate}
    \item the problem can be represented in a symmetric way, thus avoiding the need for careful choices of parametrization, and
    \item we only have to fit a group-lasso without constraints on the coefficients, which is a well-studied problem.
    \end{enumerate}
  \end{subsection}
  
  \begin{subsection}{Equivalence with unconstrained group-lasso} \label{sec:equivalence_with_group_lasso}
    We show that the overlapped group-lasso above can be solved with a simple group-lasso. We will need two Lemmas. The first shows that because we fit an intercept in the model, the estimated coefficients $\hat{\beta}$ for categorical variables will have mean zero.
    \begin{lemma}\label{lemma:zero_mean}
      Let $X$ be an indicator matrix. Then the solution $\hat{\beta}$ to
      \begin{eqnarray}
        \argmin_{\mu, \beta}\frac{1}{2}\left\|\mathbf{Y}-\mu\cdot\ind-\mathbf{X}\beta\right\|_2^2 + \lambda\|\beta\|_2
      \end{eqnarray}
      satisfies
      \begin{eqnarray}
        \bar{\hat{\beta}} = 0.
      \end{eqnarray}
      The same is true for logistic loss.
    \end{lemma}

    \begin{proof}
      Because $\mathbf{X}$ is an indicator matrix, each row consists of exactly a single 1 (all other entries 0), so that
      \begin{eqnarray}
        \mathbf{X}\cdot c\ind = c\ind
      \end{eqnarray}
      for any constant $c$. It follows that if $\hat{\mu}$ and $\hat{\beta}$ are solutions, then so are $\hat{\mu}+c\ind$ and $\hat{\beta}-c\ind$. But the norm $\|\hat{\beta}-c\ind\|_2$ is minimized for $c=\bar{\hat{\beta}}$.
    \end{proof}
    
    The next Lemma states that if we include two intercepts in the model, one penalized and the other unpenalized, then the penalized intercept will be estimated to be zero. This is because we can achieve the same fit with a lower penalty by taking $\mu\longleftarrow\mu+\tilde{\mu}$.
    \begin{lemma}\label{lemma:zero_intercept}
      The optimization problem
      \begin{eqnarray}
        \argmin_{\mu, \tilde{\mu}, \beta}\frac{1}{2}\left\|\mathbf{Y}-\mu\cdot\ind-\tilde{\mu}\cdot\ind-\ldots\right\|_2^2 + \lambda\sqrt{\|\tilde{\mu}\|_2^2+\|\beta\|_2^2}
      \end{eqnarray}
      has solution $\hat{\tilde{\mu}} = 0$ for all $\lambda>0$. The same result holds for logistic loss.
    \end{lemma}
    
    The next theorem shows how the overlapped group-lasso in Section \ref{sec:overlapped_group_lasso} reduces to a group-lasso.
    \begin{theorem}\label{theorem:main_theorem}
      Solving the constrained optimization problem (\ref{eq:overlapped_group_lasso}) - (\ref{eq:overlapped_interaction_constraints}) in Section \ref{sec:overlapped_group_lasso} is equivalent to solving the unconstrained problem
      \begin{eqnarray}
        \argmin_{\mu,\beta}\frac{1}{2}\left\|\mathbf{Y}-\mu\cdot\ind-\mathbf{X}_1\beta_1-\mathbf{X}_2\beta_2-\mathbf{X}_{1:2}\beta_{1:2}\right\|_2^2 + \lambda\left(\|\beta_1\|_2+\|\beta_2\|_2+\|\beta_{1:2}\|_2\right).
      \end{eqnarray}
    \end{theorem}

    \begin{proof}
      We need to show that the group-lasso objective can be equivalently written as an overlapped group-lasso with the appropriate constraints on the parameters. We begin by rewriting (\ref{eq:overlapped_group_lasso}) as
      \begin{multline}
        \argmin_{\mu,\tilde{\mu},\alpha,\tilde{\alpha}} \frac{1}{2}\left\|\mathbf{Y}-\mu\cdot\ind-\mathbf{X}_1\alpha_1-\mathbf{X}_2\alpha_2-[\ind\text{  }\mathbf{X}_1\text{  }\mathbf{X}_2\text{  }\mathbf{X}_{1:2}]\left[
            \begin{array}{c}
              \tilde{\mu}\\
              \tilde{\alpha}_1\\
              \tilde{\alpha}_2\\
              \alpha_{1:2}
            \end{array}
          \right]\right\|_2^2\\
        + \lambda\left(\|\alpha_1\|_2 + \|\alpha_2\|_2 + \sqrt{L_1L_2\|\tilde{\mu}\|_2^2+L_2\|\tilde{\alpha}_1\|_2^2 + L_1\|\tilde{\alpha}_2\|_2^2+\|\alpha_{1:2}\|_2^2}\right).
      \end{multline}
      By Lemma \ref{lemma:zero_intercept}, we will estimate $\hat{\tilde{\mu}} = 0$. Therefore we have not changed the solutions in any way.

      Lemma \ref{lemma:zero_mean} shows that the first two constraints in (\ref{eq:overlapped_main_effect_constraints}) are satisfied by the estimated main effects $\hat{\beta}_1$ and $\hat{\beta}_2$. We now show that
      \begin{eqnarray}
        \|\beta_{1:2}\|_2 = \sqrt{L_1L_2\|\tilde{\mu}\|_2^2+L_2\|\tilde{\alpha}_1\|_2^2 + L_1\|\tilde{\alpha}_2\|_2^2+\|\alpha_{1:2}\|_2^2}
      \end{eqnarray}
      where the $\tilde{\alpha}_1, \tilde{\alpha}_2$, and $\alpha_{1:2}$ satisfy the constraints in (\ref{eq:overlapped_main_effect_constraints}) and (\ref{eq:overlapped_interaction_constraints}).

      For fixed levels $i$ and $j$, we can decompose $\beta_{1:2}$ (see \cite{Scheffe:1959:Wiley}) as
      \begin{eqnarray}
        \beta_{1:2}^{ij}& =& \beta_{1:2}^{\cdot\cdot} + (\beta_{1:2}^{i\cdot}-\beta_{1:2}^{\cdot\cdot}) + (\beta_{1:2}^{\cdot j}-\beta_{1:2}^{\cdot\cdot}) + (\beta_{1:2}^{ij}-\beta_{1:2}^{i\cdot}-\beta_{1:2}^{\cdot j}+\beta_{1:2}^{\cdot\cdot})\\
        &\equiv& \tilde{\mu} + \tilde{\alpha}_1^i + \tilde{\alpha}_2^j + \alpha_{1:2}^{ij}.
      \end{eqnarray}
      It follows that the whole $(L_1L_2)$-vector $\beta_{1:2}$ can be written as
      \begin{eqnarray}
        \label{eq:interaction_decomposition}
        \beta_{1:2} = \ind\tilde{\mu} + \mathbf{Z}_1\tilde{\alpha}_1 + \mathbf{Z}_2\tilde{\alpha}_2 + \alpha_{1:2},
      \end{eqnarray}
      where $\mathbf{Z}_1$ is a $L_1L_2\times L_1$ indicator matrix of the form
      \begin{eqnarray}
        \underbrace{\begin{pmatrix}
            \ind_{L_2\times1} & 0 & \cdots & 0\\
            0 & \ind_{L_2\times 1} & \cdots & 0\\
            0 & 0 & \ddots & 0\\
            0 & 0 & 0 & \ind_{L_2\times1}
          \end{pmatrix}}_{L_1 \text{ columns}}
      \end{eqnarray}
      and $\mathbf{Z}_2$ is a $L_1L_2\times L_2$ indicator matrix of the form
      \begin{eqnarray}
        \text{$L_1$ copies}\left\{
          \begin{pmatrix}
            I_{L_2\times L_2}\\
            \vdots\\
            I_{L_2\times L_2}
          \end{pmatrix}
        \right.
      \end{eqnarray}
      It follows that
      \begin{eqnarray}
        \mathbf{Z}_1\tilde{\alpha}_1 = (\underbrace{\tilde{\alpha}_1^1,\ldots,\tilde{\alpha}_1^1}_{L_2 \text{ copies}}, \underbrace{\tilde{\alpha}_1^2,\ldots,\tilde{\alpha}_1^2}_{L_2 \text{ copies}},\ldots,\underbrace{\tilde{\alpha}_1^{L_1},\ldots,\tilde{\alpha}_1^{L_1}}_{L_2 \text{ copies}})^T
      \end{eqnarray}
      and
      \begin{eqnarray}
        \mathbf{Z}_2\tilde{\alpha}_2 = (\tilde{\alpha}_2^1,\ldots,\tilde{\alpha}_2^{L_2},\tilde{\alpha}_2^1,\ldots,\tilde{\alpha}_2^{L_2},\ldots, \tilde{\alpha}_2^1,\ldots,\tilde{\alpha}_2^{L_2})^T.
      \end{eqnarray}
      
      Note that $\tilde{\alpha}_1, \tilde{\alpha}_2$, and $\alpha_{1:2}$, by definition, satisfy the constraints (\ref{eq:overlapped_main_effect_constraints}) and (\ref{eq:overlapped_interaction_constraints}). This can be used to show, by direct calculation, that the four additive components in (\ref{eq:interaction_decomposition}) are mutually orthogonal, so that we can write
      \begin{eqnarray}
        \|\beta_{1:2}\|_2^2 &=& \|\ind\tilde{\mu}\|_2^2 + \|\mathbf{Z}_1\tilde{\alpha}_1\|_2^2 + \|\mathbf{Z}_2\tilde{\alpha}_2\|_2^2 + \|\alpha_{1:2}\|_2^2\\
        & = & L_1L_2\|\tilde{\mu}\|_2^2 + L_2\|\tilde{\alpha}_1\|_2^2 + L_1\|\tilde{\alpha}_2\|_2^2 + \|\alpha_{1:2}\|_2^2.
      \end{eqnarray}
      
      We have shown that the penalty in the group-lasso problem is equivalent to the penalty in the constrained overlapped group-lasso. It remains to show that the loss functions in both problems are also the same. Since $\mathbf{X}_{1:2}\mathbf{Z}_1=\mathbf{X}_1$ and $\mathbf{X}_{1:2}\mathbf{Z}_2=\mathbf{X}_2$, this can be seen by a direct computation:
      \begin{eqnarray}
        \mathbf{X}_{1:2}\beta_{1:2} & = & \mathbf{X}_{1:2}(\ind\tilde{\mu} + \mathbf{Z}_1\tilde{\alpha}_1 + \mathbf{Z}_2\tilde{\alpha}_2 + \alpha_{1:2})\\
        & = & \ind\tilde{\mu} + \mathbf{X}_1\tilde{\alpha}_1 + \mathbf{X}_2\tilde{\alpha}_2 + \mathbf{X}_{1:2}\alpha_{1:2}\\
        & = & [\ind\text{  }\mathbf{X}_1\text{  }\mathbf{X}_2\text{  }\mathbf{X}_{1:2}]\left[
          \begin{array}{c}
            \tilde{\mu}\\
            \tilde{\alpha}_1\\
            \tilde{\alpha}_2\\
            \alpha_{1:2}
          \end{array}
          \right]
      \end{eqnarray}
    \end{proof}

    Theorem \ref{theorem:main_theorem} shows that we can use the group-lasso to obtain estimates that satisfy strong hierarchy, without solving the overlapped group-lasso with constraints. The theorem also shows that the main effects and interactions can be extracted with
    \begin{eqnarray}
      \hat{\theta}_1 &=& \hat{\beta}_1 + \hat{\tilde{\alpha}}_1\\
      \hat{\theta}_2 &=& \hat{\beta}_2 + \hat{\tilde{\alpha}}_2\\
      \hat{\theta}_{1:2} &=& \hat{\alpha}_{1:2}.
    \end{eqnarray}
We discuss the properties of the \textsc{glinternet} estimates in the next section.
  \end{subsection}

  \begin{subsection}{Properties of the \textsc{glinternet} estimators}\label{sec:properties_of_the_glinternet_estimators}
    While \textsc{glinternet} treats the problem as a group-lasso, examining the equivalent overlapped group-lasso version makes it easier to draw insights about the behaviour of the method under various scenarios. Recall that the overlapped penalty for two variables is given by
    \begin{eqnarray}
      \|\alpha_1\|_2 + \|\alpha_2\|_2 + \sqrt{L_2\|\tilde{\alpha}_1\|_2^2 + L_1\|\tilde{\alpha}_2\|_2^2 + \|\alpha_{1:2}\|_2^2}.
    \end{eqnarray}
    If the ground truth is additive, i.e. $\alpha_{1:2}=0$, then $\tilde{\alpha}_1$ and $\tilde{\alpha}_2$ will be estimated to be zero. This is because for $L_1,L_2\geq2$ and $a,b\geq0$, we have
    \begin{eqnarray}
      \sqrt{L_2a^2+L_1b^2}\geq a+b.
    \end{eqnarray}
    Thus it is advantageous to place all the main effects in $\alpha_1$ and $\alpha_2$, because doing so results in a smaller penalty. Therefore, if the truth has no interactions, then \textsc{glinternet} picks out only main effects.

    If an interaction was present ($\alpha_{1:2}>0$), the derivative of the penalty term with respect to $\alpha_{1:2}$ is
    \begin{eqnarray}
      \frac{\alpha_{1:2}}{\sqrt{L_2\|\tilde{\alpha}_1\|_2^2 + L_1\|\tilde{\alpha}_2\|_2^2 + \|\alpha_{1:2}\|_2^2}}.
    \end{eqnarray}
    The presence of main effects allows this derivative to be smaller, thus allowing the algorithm to pay a smaller penalty (as compared to no main effects present) for making $\hat{\alpha}_{1:2}$ nonzero. This shows interactions whose main effects are also present are discovered before pure interactions.
  \end{subsection}
  
  \begin{subsection}{Interaction between a categorical variable and a continuous variable}
    We describe how to extend Theorem \ref{theorem:main_theorem} to interaction between a continuous variable and a categorical variable.
    
    Consider the case where we have a categorical variable $F$ with $L$ levels, and a continuous variable $Z$. Let $\mu_i=\E[Y|F=i,Z=z]$. There are four cases:
    \begin{itemize}
    \item $\mu_i = \mu$ (no main effects, no interactions)
    \item $\mu_i = \mu + \theta_1^i$ (main effect $F$)
    \item $\mu_i = \mu + \theta_1^i + \theta_2z$ (two main effects)
    \item $\mu_i = \mu + \theta_1^i + \theta_2z + \theta_{1:2}^iz$ (main effects and interaction)
    \end{itemize}
    As before, we impose the constraints $\sum_{i=1}^L\theta_1^i=0$ and $\sum_{i=1}^L\theta_{1:2}^i=0$. An overlapped group-lasso of the form
    \begin{multline}
      \label{eq:overlapped_group_lasso_categorical_continuous}
      \argmin_{\mu, \alpha, \tilde{\alpha}} \frac{1}{2}\left\|\mathbf{Y}-\mu\cdot\ind-\mathbf{X}\alpha_1-\mathbf{Z}\alpha_2-[\mathbf{X}\quad\mathbf{Z}\quad(\mathbf{X}*\mathbf{Z})]\left[
          \begin{array}{c}
            \tilde{\alpha}_1\\
            \tilde{\alpha}_2\\
            \alpha_{1:2}
          \end{array}
        \right]\right\|_2^2\\
      + \lambda\left(\|\alpha_1\|_2 + \|\alpha_2\|_2 + \sqrt{\|\tilde{\alpha}_1\|_2^2 + L\|\tilde{\alpha}_2\|_2^2+\|\alpha_{1:2}\|_2^2}\right)
    \end{multline}
    subject to
    \begin{eqnarray}
      \sum_{i=1}^{L}\alpha_1^i=0,\quad  \sum_{i=1}^{L}\tilde{\alpha}_1^i = 0,\quad \sum_{i=1}^{L}\alpha_{1:2}^i = 0
    \end{eqnarray}
    allows us to obtain estimates of the interaction term that satisfy strong hierarchy. This is again due to the nature of the square root term in the penalty. The actual main effects and interactions can be recovered as
    \begin{eqnarray}
      \hat{\theta}_1 = \hat{\alpha}_1 + \hat{\tilde{\alpha}}_1\\
      \hat{\theta}_2 = \hat{\alpha}_2 + \hat{\tilde{\alpha}}_2\\
      \hat{\theta}_{1:2} = \hat{\alpha}_{1:2}.
    \end{eqnarray}
    
    We have the following extension of Theorem \ref{theorem:main_theorem}:
    \begin{theorem}
      \label{theorem:main_theorem_extension_1}
      Solving the constrained overlapped group-lasso above is equivalent to solving
      \begin{eqnarray}
        \argmin_{\mu,\beta}\frac{1}{2}\left\|\mathbf{Y}-\mu\cdot\ind-\mathbf{X}\beta_1-\mathbf{Z}\beta_2-(\mathbf{X}*[\ind\quad\mathbf{Z}])\beta_{1:2}\right\|_2^2 + \lambda\left(\|\beta_1\|_2+\|\beta_2\|_2+\|\beta_{1:2}\|_2\right).
      \end{eqnarray}
    \end{theorem}
    
    \begin{proof}
      We proceed as in the proof of Theorem \ref{theorem:main_theorem} and introduce an additional parameter $\tilde{\mu}$ into the overlapped objective:
      \begin{multline}
        \argmin_{\mu, \tilde{\mu}, \alpha, \tilde{\alpha}} \frac{1}{2}\left\|\mathbf{Y}-\mu\cdot\ind-\mathbf{X}\alpha_1-\mathbf{Z}\alpha_2-[\ind\quad \mathbf{X}\quad\mathbf{Z}\quad(\mathbf{X}*\mathbf{Z})]\left[
            \begin{array}{c}
              \tilde{\mu}\\
              \tilde{\alpha}_1\\
              \tilde{\alpha}_2\\
              \alpha_{1:2}
            \end{array}
          \right]\right\|_2^2\\
        + \lambda\left(\|\alpha_1\|_2 + \|\alpha_2\|_2 + \sqrt{L\|\tilde{\mu}\|_2^2+\|\tilde{\alpha}_1\|_2^2 + L\|\tilde{\alpha}_2\|_2^2+\|\alpha_{1:2}\|_2^2}\right)
      \end{multline}
      As before, this does not change the solutions because we will have $\hat{\tilde{\mu}}= 0$ (see Lemma \ref{lemma:zero_intercept}).
      
      Decompose the $2L$-vector $\beta_{1:2}$ into
      \begin{eqnarray}
        \left[
          \begin{array}{c}
            \eta_1\\
            \eta_2
          \end{array}
        \right],
      \end{eqnarray}
      where $\eta_1$ and $\eta_2$ both have dimension $L\times1$. Apply the anova decomposition to both to obtain
      \begin{eqnarray}
        \eta_1^i &=& \eta_1^\cdot + (\eta_1^i - \eta_1^\cdot)\\
        &\equiv& \tilde{\mu} + \tilde{\alpha}_1^i
      \end{eqnarray}
      and
      \begin{eqnarray}
        \eta_2^i& = & \eta_2^\cdot + (\eta_2^i-\eta_2^\cdot)\\
        &\equiv& \tilde{\alpha}_2 + \alpha_{1:2}^i.
      \end{eqnarray}
      Note that $\tilde{\alpha}_1$ is a $(L\times1)$-vector that satisfies $\sum_{i=1}^L\tilde{\alpha}_2^i=0$, and likewise for $\alpha_{1:2}$. This allows us to write
      \begin{eqnarray}
        \beta_{1:2} = \left[
          \begin{array}{c}
            \tilde{\mu}\cdot\ind_{L\times1}\\
            \tilde{\alpha}_2\cdot\ind_{L\times1}
          \end{array}
          \right] + \left[
          \begin{array}{c}
            \tilde{\alpha}_1\\
            \alpha_{1:2}
          \end{array}
          \right]
      \end{eqnarray}
      It follows that
      \begin{eqnarray}
        \|\beta_{1:2}\|_2^2 = L\|\tilde{\mu}\|_2^2 + \|\tilde{\alpha}_1\|_2^2 + L\|\tilde{\alpha}_2\|_2^2 + \|\alpha_{1:2}\|_2^2,
      \end{eqnarray}
      which shows that the penalties in both problems are equivalent. A direct computation shows that the loss functions are also equivalent:
      \begin{eqnarray}
        (\mathbf{X}*[\ind\quad\mathbf{Z}])\beta_{1:2} &=& [\mathbf{X}\quad(\mathbf{X}*\mathbf{Z})]\beta_{1:2}\\
        & =& [\mathbf{X} \quad(\mathbf{X}*\mathbf{Z})]\left(\left[ \begin{array}{c}
              \tilde{\mu}\cdot\ind_{L\times1}\\
              \tilde{\alpha}_2\cdot\ind_{L\times1}
            \end{array}
          \right] + \left[
            \begin{array}{c}
              \tilde{\alpha}_1\\
              \alpha_{1:2}
            \end{array}
          \right]
          \right)\\
          & = & \tilde{\mu}\cdot\ind + \mathbf{X}\tilde{\alpha}_1 + \mathbf{Z}\tilde{\alpha}_2 + (\mathbf{X}*\mathbf{Z})\alpha_{1:2}.
      \end{eqnarray}
    \end{proof}
    
    Theorem \ref{theorem:main_theorem_extension_1} allows us to accomodate interactions between continuous and categorical variables by simply parametrizing the interaction term as $\mathbf{X} *[\ind\quad\mathbf{Z}]$, where $\mathbf{X}$ is the indicator matrix representation for categorical variables that we have been using all along. We then proceed as before with a group-lasso.
  \end{subsection}

  \begin{subsection}{Interaction between two continuous variables}
    We have seen that the appropriate representations for the interaction terms are
    \begin{itemize}
    \item $\mathbf{X}_1*\mathbf{X}_2 = \mathbf{X}_{1:2}$ for categorical variables
    \item $\mathbf{X}*[\ind\quad \mathbf{Z}]=[\mathbf{X} \quad(\mathbf{X}*\mathbf{Z})]$ for one categorical variable and one continuous variable.
    \end{itemize}
    How should we represent the interaction between two continuous variables? Let $Z_1$ and $Z_2$ be two continuous variables. One might guess by now that the appropriate form of the interaction term is given by
    \begin{eqnarray}
      \mathbf{Z}_{1:2} &=& [\ind\quad \mathbf{Z}_1]*[\ind\quad \mathbf{Z}_2]\\
      & = & [\ind\quad \mathbf{Z}_1\quad\mathbf{Z}_2\quad(\mathbf{Z}_1*\mathbf{Z}_2)].
    \end{eqnarray}
    This is indeed the case. A linear interaction model for $\mathbf{Z}_1$ and $\mathbf{Z}_2$ is given by
    \begin{eqnarray}
      \E[Y|Z_1=z_1,Z_2=z_2] = \mu + \theta_1z_1 + \theta_2z_2 + \theta_{1:2}z_1z_2.
    \end{eqnarray}
    Unlike the previous cases where there were categorical variables, there are no constraints on any of the coefficients. It follows that the overlapped group-lasso
    \begin{multline}
      \argmin_{\mu,\tilde{\mu},\alpha,\tilde{\alpha}}\frac{1}{2}\left\|\mathbf{Y}-\mu\cdot\ind-\mathbf{Z}_1\alpha_1-\mathbf{Z}_2\alpha_2-[\ind\quad\mathbf{Z}_1\quad\mathbf{Z}_2\quad(\mathbf{Z}_1*\mathbf{Z}_2)]\left[
        \begin{array}{c}
          \tilde{\mu}\\
          \tilde{\alpha}_1\\
          \tilde{\alpha}_2\\
          \alpha_{1:2}
        \end{array}
        \right]\right\|_2^2\\
        + \lambda\left(\|\alpha_1\|_2+\|\alpha_2\|_2+\sqrt{\|\tilde{\mu}\|_2^2+\|\tilde{\alpha}_1\|_2^2+\|\tilde{\alpha}_2\|_2^2+\|\alpha_{1:2}\|_2^2}\right)
    \end{multline}
    is trivially equivalent to
    \begin{eqnarray}
      \argmin_{\mu,\beta}\frac{1}{2}\left\|\mathbf{Y}-\mu\cdot\ind-\mathbf{Z}_1\beta_1-\mathbf{Z}_2\beta_2-([\ind\quad\mathbf{Z}_1]*[\ind\quad\mathbf{Z}_2])\beta_{1:2}\right\|_2^2 + \lambda\left(\|\beta_1\|_2+\|\beta_2\|_2+\|\beta_{1:2}\|_2\right),
    \end{eqnarray}
    with the $\beta$'s taking the place of the $\alpha$'s. Note that we will have $\hat{\tilde{\mu}}=0$.
  \end{subsection}
\end{section}

\begin{section}{Variable screening}\label{sec:screening}
  \textsc{glinternet} works by solving a group-lasso with $p+
  \begin{pmatrix}
    p\\
    2
  \end{pmatrix}$
  groups of variables. Even for moderate $p$ ($\sim 10^5$), we will require some form of screening to reduce the dimension of the interaction search space. We have argued that models satisfying hierarchy make sense, so that it is natural to consider screening devices that hedge on the presence of main effects. We discuss two screening methods in this section: gradient boosting, and an adaptive screen based on the strong rules of \cite{Tibshirani:2012:JRSS}. We describe the boosting approach first.

  \begin{subsection}{Screening with boosted trees}\label{sec:screening_with_boosted_trees}
    AdaBoost \cite{Freund:1995:DGO} and gradient boosting \cite {Friedman:2001} are effective approaches to building ensembles of weak learners such as decision trees. One of the advantages of trees is that they are able to model nonlinear effects and high-order interactions. For example, a depth-2 tree essentially represents an interaction between the variables involved in the two splits, which suggests that boosting with depth-2 trees is a way of building a first-order interaction model. Note that the interactions are hierarchical, because in finding the optimal first split, the boosting algorithm is looking for the best main effect. The subsequent split is then made, conditioned on the first split.

    If we boost with $T$ trees, then we end up with a model that has at most $T$ interaction pairs. The following diagram gives a schematic of the boosting iterations with categorical variables.

    \begin{multline}
      \xymatrix{
        &&*+[o][F-]{F_1}\ar[dl]_{\{1\}}\ar[dr]^{\{2,3\}}&\\
        &*+[o][F-]{F_2}\ar[dl]_{\{1,2\}}\ar[dr]^{\{3\}}&&*+[F-]{}\\
        *+[F-]{}&&*+[F-]{}&&
      } + \xymatrix{
        &&*+[o][F-]{F_{11}}\ar[dl]_{\{2\}}\ar[dr]^{\{1,3\}}&\\
        &*+[o][F-]{F_{23}}\ar[dl]_{\{3\}}\ar[dr]^{\{1,2\}}&&*+[F-]{}\\
        *+[F-]{}&&*+[F-]{}&&
      }
    \end{multline}

    In the first tree, levels 2 and 3 of $F_1$ are not involved in the interaction with $F_2$. Therefore each tree in the boosted model does not represent an interaction among all the levels of the two variables, but only among a subset of the levels. To enforce the full interaction structure, one could use fully-split trees, but we do not develop this approach for two reasons. First, boosting is a sequential procedure and is quite slow even for moderately sized problems. Using fully split trees will further degrade its runtime. Second, in variables with many levels, it is reasonable to expect that the interactions only occur among a few of the levels. If this were true, then a complete interaction that is weak for every combination of levels might be selected over a strong partial interaction. But it is the strong partial interaction that we are interested in.
    
    Boosting is feasible because it is a greedy algorithm. If $p$ is the number of variables, an exhaustive search involves $\mathcal{O}(p^2)$ variables, whereas boosting operates with $\mathcal{O}(p)$. To use the boosted model as a screening device for interaction candidates, we take the set of all unique interactions from the collection of trees. For example, in our schematic above, we would add $F_{1:2}$ and $F_{11:23}$ to our candidate set of interactions.

    In our experiments, using boosting as a screen did not perform as well as we hoped. There is the issue of selecting tuning parameters such as the amount of shrinkage and the number of trees to use. Lowering the shrinkage and increasing the number of trees improves false discovery rates, but at a significant cost to speed. In the next section, we describe a screening approach that is based on computing inner products that is efficient and that can be integrated with the strong rules for the group lasso.
  \end{subsection}
  
  \begin{subsection}{An adaptive screening procedure}\label{sec:subsection_adaptive_screening}
    The strong rules \cite{Tibshirani:2012:JRSS} for lasso-type problems are effective heuristics for discarding large numbers of variables that are likely to be redundant. As a result, the strong rules can dramatically speed up the convergence of algorithms because they can concentrate on a smaller set (we call this the \textit{strong set}) of variables that are more likely to be nonzero. The strong rules are not safe, however, meaning that it is possible that some of the discarded variables are actually supposed to be nonzero. Because of this, after our algorithm has converged on the strong set, we have to check the KKT conditions on the discarded set. Those variables that do not satisfy the conditions then have to be added to the current set of nonzero variables, and we fit on this expanded set. This happens rarely in our experience, i.e. the discarded variables tend to remain zero after the algorithm has converged on the strong set, which means we rarely have to do multiple rounds of fitting for any given value of the regularization paramter $\lambda$.

    The strong rules for the group-lasso involve computing $s_i=\|\mathbf{X}_i^T(\mathbf{Y}-\hat{\mathbf{Y}})\|_2$ for every group of variables $\mathbf{X}_i$, and then discarding a group $i$ if $s_i<2\lambda_{current}-\lambda_{previous}$. If this is feasible for all $p+
    \begin{pmatrix}
      p\\
      2
    \end{pmatrix}$
    groups, then there is no need for screening; we simply fit the group-lasso on those groups that pass the strong rules filter. Otherwise, we approximate this by screening only on the groups that correspond to main effects. We then take the candidate set of interactions to consist of all pairwise interactions between the variables that passed this screen. Note that because the KKT conditions for group $i$ are (see Section \ref{sec:group_lasso})
    \begin{eqnarray}
      s_i < \lambda & \text{if} \quad \hat{\beta}_i=0\\
      s_i = \lambda & \text{if} \quad \hat{\beta}_i\neq0,
    \end{eqnarray}
    we will have already computed the $s_i$ for the strong rules from checking the KKT conditions for the solutions at the previous $\lambda$. This allows us to integrate screening with the strong rules in an efficient manner. An example will illustrate.

    Suppose we have 10,000 variables ($\sim50\times10^6$ possible interactions), but we are computationally limited to a group-lasso with $10^6$ groups. Assume we have the fit for $\lambda=\lambda_k$, and want to move on to $\lambda_{k+1}$. Let $r_{\lambda_k}=\mathbf{Y}-\hat{\mathbf{Y}}_{\lambda_k}$ denote the current residual. At this point, the variable scores $s_i=\|\mathbf{X}_i^Tr_{\lambda_k}\|_2$ have already been computed from checking the KKT conditions at the solutions for $\lambda_k$. We restrict ourselves to the 10,000 variables, and take the 100 with the highest scores. Denote this set by $\mathcal{T}_{100}^{\lambda_{k+1}}$. The candidate set of variables for the group-lasso is then given by $\mathcal{T}_{100}^{\lambda_{k+1}}$ together with the pairwise interactions between \textit{all} 10,000 variables and $\mathcal{T}_{100}^{\lambda_{k+1}}$. Because this gives a candidate set with about $100\times 10,000=10^6$ terms, the compuation is now feasible. We then compute the group-lasso on this candidate set, and repeat the procudure with the new residual $r_{\lambda_k+1}$.

    This screen is easy to compute since it is based on inner products. Moreover, they can be computed in parallel. The procedure also integrates well with the strong rules by reusing inner products computed from the fit for a previous $\lambda$.
  \end{subsection}
\end{section}

\begin{section}{Related work and approaches}\label{sec:related_work}
  We describe some past and related approaches to discovering interactions. We give a short synopsis of how they work, and say why they are inadequate for our purposes. The method most similar to ours is hierNet.

  \begin{subsection}{Logic regression \cite{Kooperberg:2003:JCGS}}
    Logic regression finds boolean combinations of variables that have high predictive power of the response variable. For example, a combination might look like
    \begin{eqnarray}
      (F_1 \text{ and } F_3) \text{ or } F_5.
    \end{eqnarray}
    This is an example of an interaction that is of higher-order than what \textsc{glinternet} handles, and is an appealing aspect of logic regression. However, logic regression does not accomodate continuous variables or categorical variables with more than two levels. We do not make comparisons with logic regression in our simulations for this reason.
  \end{subsection}

  \begin{subsection}{Composite absolute penalties \cite{Yu:2009:AnnalsStat}}
    Like \textsc{glinternet}, this is also a penalty-based approach. CAP employs penalties of the form
    \begin{eqnarray}
      \|(\beta_i,\beta_j)\|_{\gamma_1} + \|\beta_j\|_{\gamma_2}
    \end{eqnarray}
    where $\gamma_1 > 1$. Such a penalty ensures that $\hat{\beta}_i\neq0$ whenever $\hat{\beta}_j\neq0$. It is possible that $\hat{\beta}_i\neq0$ but $\hat{\beta}_j=0$. In other words, the penalty makes $\hat{\beta}_j$ hierarchically dependent on $\hat{\beta}_i$: it can only be nonzero after $\hat{\beta}_i$ becomes nonzero. It is thus possible to use CAP penalties to build interaction models that satisfy hierarchy. For example, a penalty of the form $\|(\theta_1,\theta_2,\theta_{1:2})\|_2 + \|\theta_{1:2}\|_2$ will result in estimates that satisfy $\hat{\theta}_{1:2}\neq0\Longrightarrow\hat{\theta}_1\neq0\text{ and }\hat{\theta}_2\neq0$. We can thus build a linear interaction model for two categorical variables by solving
    \begin{eqnarray}
      \argmin_{\mu,\theta}\frac{1}{2}\left\|\mathbf{Y}-\mu\cdot\ind-\mathbf{X}_1\theta_1-\mathbf{X}_2\theta_2-\mathbf{X}_{1:2}\theta_{1:2}\right\|_2^2 + \lambda(\|(\theta_1,\theta_2,\theta_{1:2})\|_2 + \|\theta_{1:2}\|_2)
    \end{eqnarray}
    subject to (\ref{eq:main_effect_constraints}) and (\ref{eq:interaction_constraints}). We see that the CAP approach differs from \textsc{glinternet} in that we have to solve a constrained optimization problem which is considerably more complicated, thus making it unclear if CAP will be computationally feasible for larger problems. The form of the penalties are also different: the interaction coefficient in CAP is penalized twice, whereas \textsc{glinternet} penalizes it once. It is not obvious what the relationship between the two algorithms' solutions would be.
  \end{subsection}

  \begin{subsection}{hierNet \cite{Bien:2013:AnnalsStat}}
    This is a method that, like \textsc{glinternet}, seeks to find interaction estimates that obey hierarchy with regularization. The optimization problem that hierNet solves is
    \begin{eqnarray}\label{eq:hierNet}
      \argmin_{\mu,\beta,\theta}\frac{1}{2}\sum_{i=1}^n(y_i-\mu-x_i^T\beta-\frac{1}{2}x_i^T\theta x_i)^2 + \lambda\ind^T(\beta^++\beta^-) +\frac{\lambda}{2}\|\theta\|_1
    \end{eqnarray}
    subject to
    \begin{eqnarray}
      \theta=\theta^T, \|\theta_j\|_1\leq\beta_j^++\beta_j^-, \beta_j^+\geq0, \beta_j^-\geq0.
    \end{eqnarray}
    The main effects are represented by $\beta$, and interactions are given by $\theta$. The first constraint enforces symmetry in the interaction coefficients. $\beta_j^+$ and $\beta_j^-$ are the positive and negative parts of $\beta_j$, and are given by $\beta_j^+=\max(0, \beta_j)$ and $\beta_j^-=-\min(0, \beta_j)$ respectively. The constraint $\|\theta_j\|_1\leq\beta_j^++\beta_j^-$ implies that if some components of the $j$-th row of $\theta$ are estimated to be nonzero, then the main effect $\beta_j$ will also be estimated to be nonzero. Since $\theta_j$ corresponds to interactions between the $j$-th variable and all the other variables, this implies that the solutions to the hierNet objective satisfy weak hierarchy. One can think of $\beta_j^+ + \beta_j^-$ as a budget for the amount of interactions that are allowed to be nonzero.

    The hierNet objective can be modified to obtain solutions that satisfy strong hierarchy, which makes it in principle comparable to \textsc{glinternet}. Currently, hierNet is only able to accomodate binary and continuous variables, and is practically limited to fitting models with fewer than 1000 variables.
  \end{subsection}
\end{section}

\begin{section}{Simulation study}\label{sec:simulation_study}
  We perform simulations to see if \textsc{glinternet} is competitive with existing methods. hierNet is a natural benchmark because it also tries to find interactions subject to hierarchical constraints. Because hierNet only works with continuous variables and 2-level categorical variables, we include gradient boosting as a competitor for the scenarios where hierNet cannot be used.
  \begin{subsection}{False discovery rates}\label{sec:false_discovery_rates}
    We simulate 4 different setups:
    \begin{enumerate}
    \item Truth obeys strong hierarchy. The interactions are only among pairs of nonzero main effects.
    \item Truth obeys weak hierarchy. Each interaction has only one of its main effects present.
    \item Truth is anti-hierarchical. The interactions are only among pairs of main effects that are not present.
    \item Truth is pure interaction. There are no main effects present, only interactions.
    \end{enumerate}
    Each case is generated with $n=500$ observations and $p=30$ continuous variables, with a signal to noise ratio of 1. Where applicable, there are 10 main effects and/or 10 interactions in the ground truth. The interaction and main effect coefficients are sampled from $N(0,1)$, so that the variance in the observations should be split equally between main effects and interactions.

  Boosting is done with 5000 depth-2 trees and a learning rate of 0.001. Each tree represents a candidate interaction, and we can compute the improvement to fit due to this candidate pair. Summing up the improvement over the 5000 trees gives a score for each interaction pair, which can then be used to order the pairs. We then compute the false discovery rate as a function of rank. For \textsc{glinternet} and hierNet, we obtain a path of solutions and compute the false discovery rate as a function of the number of interactions discovered. The default setting for hierNet is to impose weak hierarchy, and we use this except in the cases where the ground truth has strong hierarchy. In these cases, we set hierNet to impose strong hierarchy. We also set ``diagonal=FALSE'' to disable quadratic terms.
  
  We plot the average false discovery rate with standard error bars as a function of the number of predicted interactions in Figure \ref{fig:simulation_results_continuous}. The results are from 100 simulation runs. We see that \textsc{glinternet} is competitive with hierNet when the truth obeys strong or weak hierarchy, and does better when the truth is anti-hierarchical. This is expected because hierNet requires the presence of main effects as a budget for interactions, whereas \textsc{glinternet} can still esimate an interaction to be nonzero even though none of its main effects are present. Boosting is not competitive, especially in the anti-hierarchical case. This is because the first split in a tree is effectively looking for a main effect.\\
  \begin{figure}
    \begin{center}
      \centerline{\includegraphics[height=3in, width=3in]{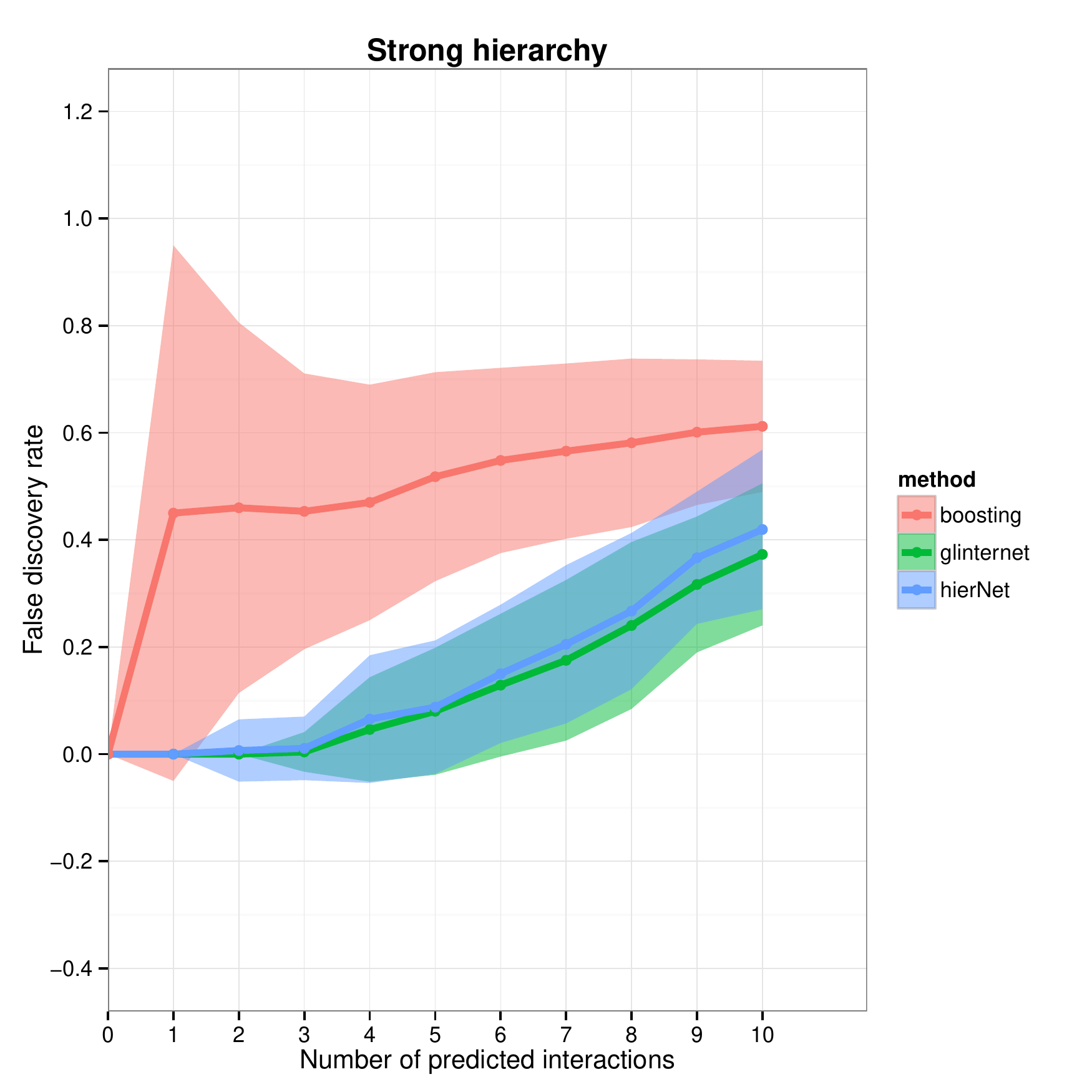}
        \includegraphics[height=3in, width=3in]{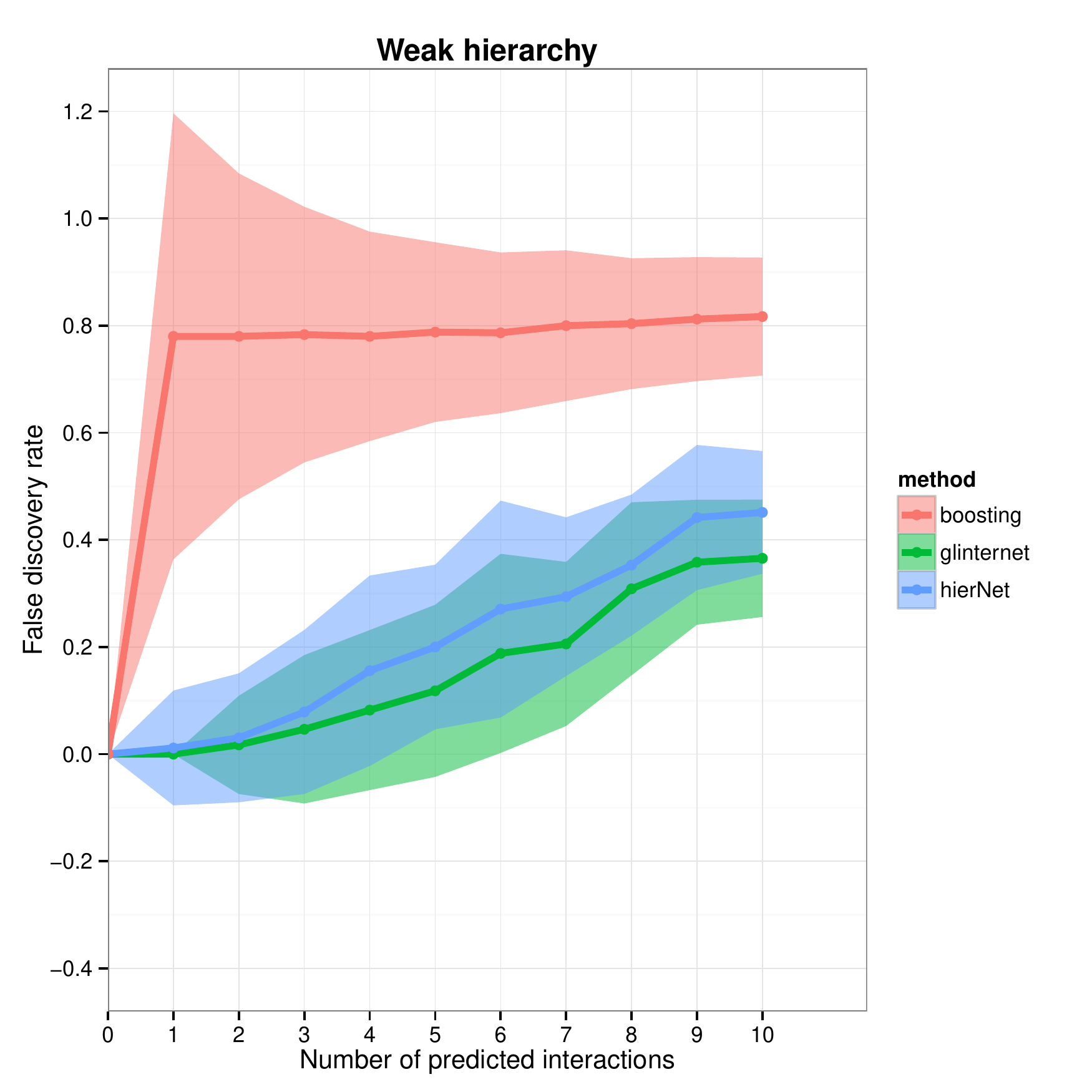}}
      \centerline{\includegraphics[height=3in, width=3in]{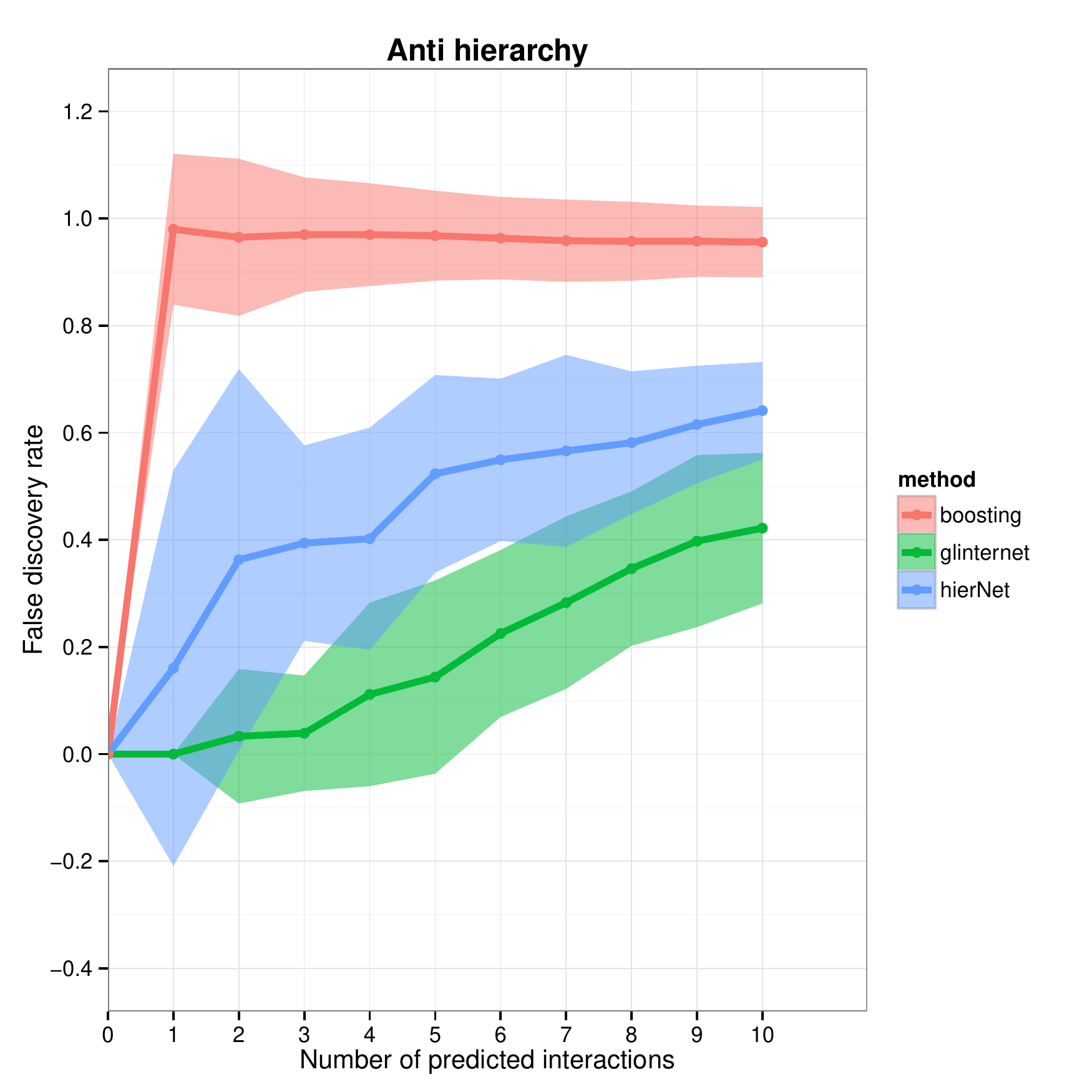}
        \includegraphics[height=3in, width=3in]{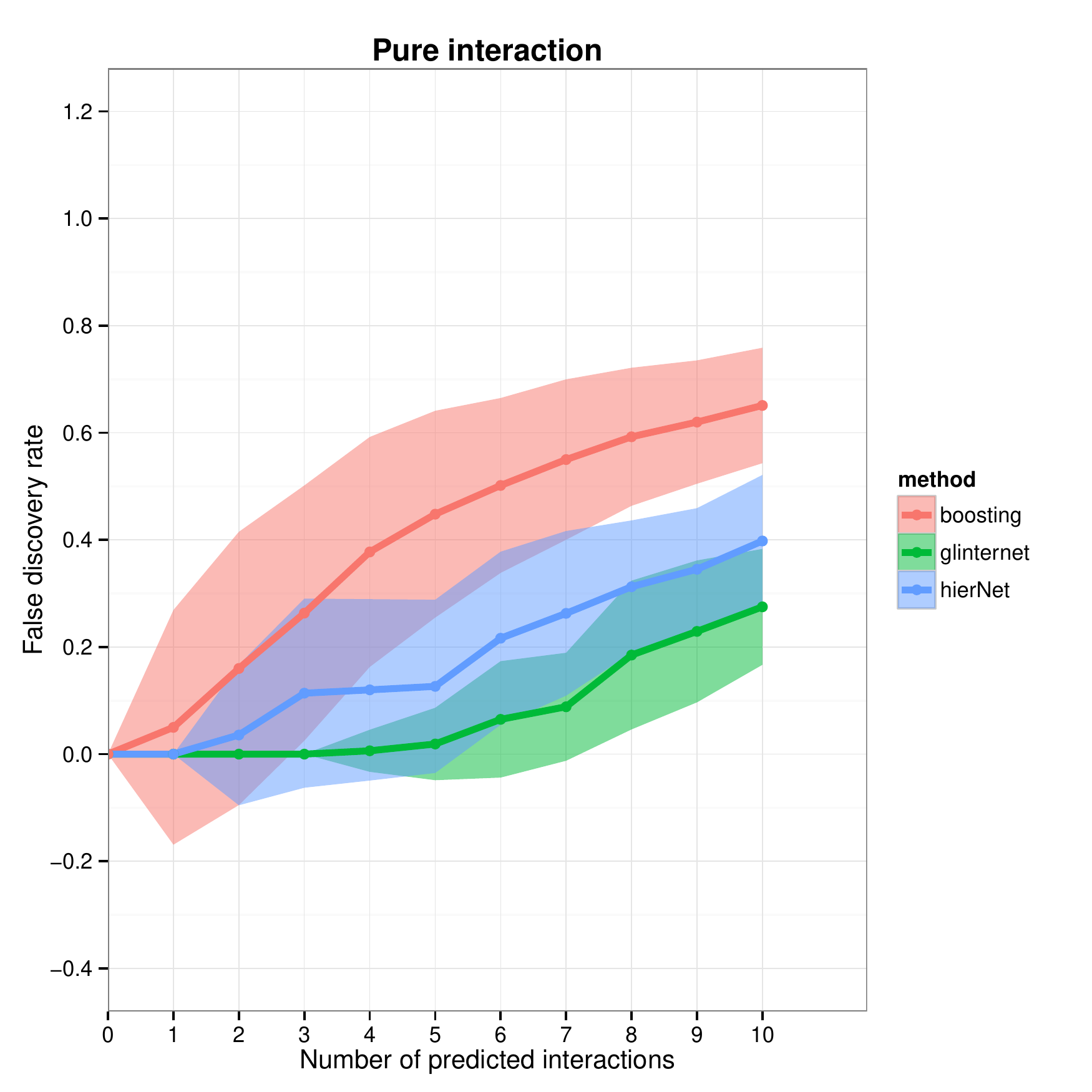}}
    \end{center}
    \caption{Simulation results for continuous variables: Average false discovery rate and standard errors from 100 simulation runs.}
    \label{fig:simulation_results_continuous}
  \end{figure}\\
  Both \textsc{glinternet} and hierNet perform comparably in these simulations. If all the variables are continuous, there do not seem to be compelling reasons to choose one over the other.
  \end{subsection}

  \begin{subsection}{Feasibility}\label{sec:feasibility}
    To the best of our knowledge, hierNet is the only readily available package for learning interactions among continuous variables in a hierarchical manner. Therefore it is natural to use hierNet as a speed benchmark. We generate data in which the ground truth has strong hierarchy as in Section \ref{sec:false_discovery_rates}, but with $n=1000$ quantitative observations and $p=20, 40, 80, 160, 320, 640$ continuous variables. We set each method to find 10 interactions. While hierNet does not allow the user to specify the number of interactions to discover, we get around this by fitting a path of values, then selecting the regularization parameter that corresponds to 10 nonzero estimated interactions. We then refit hierNet along a path that terminates with this choice of parameter, and time this run. Both software packages are compiled with the same options. Figure \ref{fig:comparison_timing} shows the best time recorded for each method over 10 runs. These simulations were timed on a Intel Core-i7 3930K processor.
  \begin{figure}
    \begin{center}
      \centerline{\includegraphics[height=3in, width=3in]{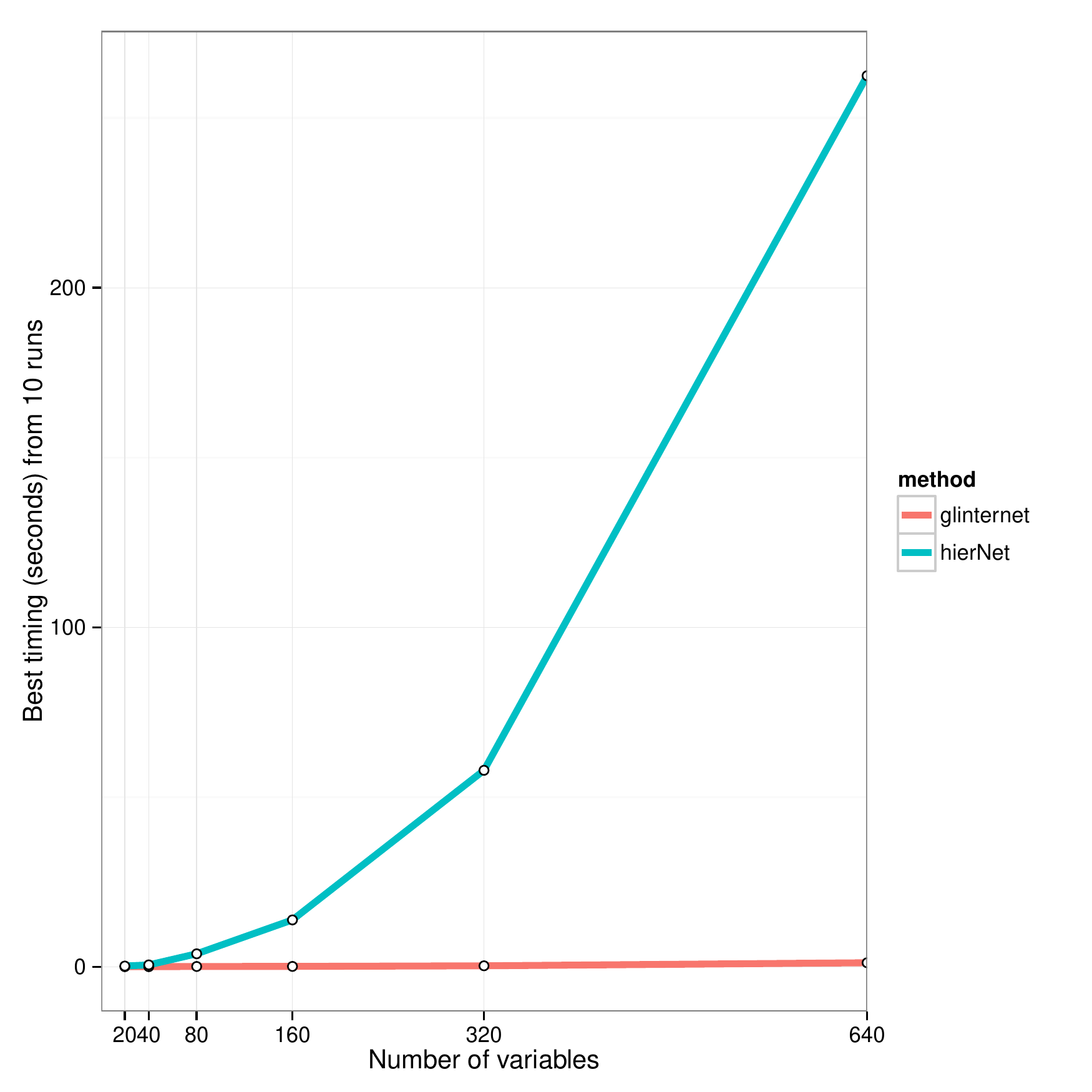}
        \includegraphics[height=3in, width=3in]{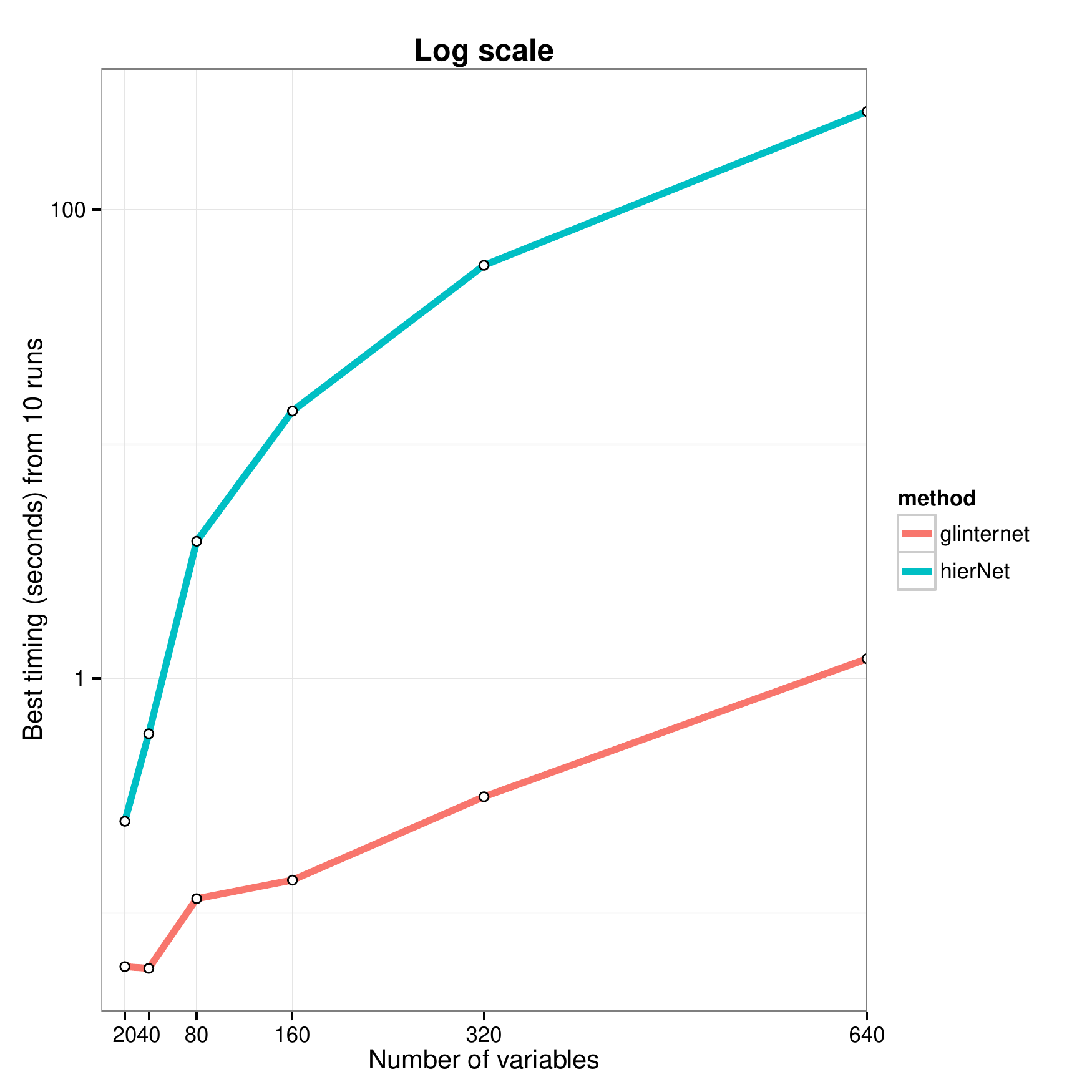}}
      \end{center}
      \caption{\textbf{Left:} Best wallclock time over 10 runs for discovering 10 interactions. \textbf{Right:} log scale.}
      \label{fig:comparison_timing}
    \end{figure}
  \end{subsection}
\end{section}

\begin{section}{Real data examples}\label{sec:real_data}
  We compare the performance of \textsc{glinternet} on several prediction problems. The competitor methods used are gradient boosting, lasso, ridge regression, and hierNet where feasible. In all situations, we determine the number of trees in boosting by first building a model with a large number of trees, typically 5000 or 10000, and then selecting the number that gives the lowest cross-validated error. We use a learning rate of 0.001, and we do not subsample the data since the sample sizes are small in all the cases.

  The methods are evaluated on three measures:
  \begin{enumerate}
  \item missclassification error, or 0-1 loss
  \item area under the receiver operating characteristic (ROC) curve, or auc
  \item cross entropy, given by $-\frac{1}{n}\sum_{i=1}^n\left[y_i\log(\hat{y}_i) + (1-y_i)\log(1-\hat{y}_i)\right]$.
  \end{enumerate}

  \begin{subsection}{South African heart disease data}
    The data consists of 462 males from a high risk region for heart disease in South Africa. The task is to predict which subjects had coronary heart disease using risk factors such as cumulative tobacco, blood pressure, and family history of heart disease. We randomly split the data into 362-100 train-test examples, and tuned each method on the training data using 10-fold cross validation before comparing the prediction performance on the held out test data. This splitting process was carried out 20 times, and Figure \ref{fig:comparison_saheart} summarizes the results. The methods are all comparable, with no distinct winner.
    \begin{figure}
      \begin{center}
        \includegraphics[height=3in, width=3in]{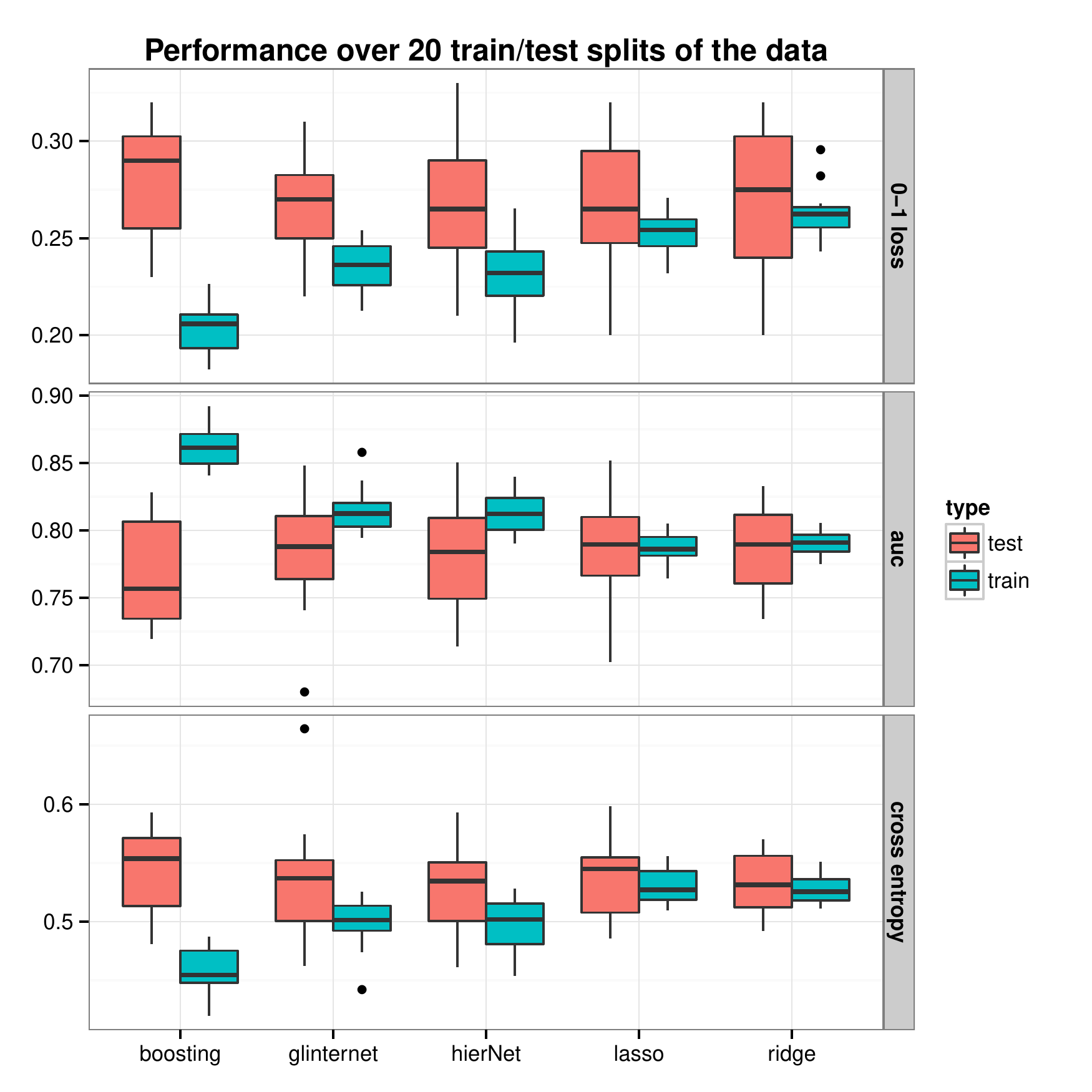}
      \end{center}
      \caption{Performance of methods on 20 train-test splits of the South African heart disease data.}
      \label{fig:comparison_saheart}
    \end{figure}
  \end{subsection}
  \begin{subsection}{Spambase}
    This is the Spambase data taken from the UCI Machine Learning Repository. There are 4601 binary observations indicating whether an email is spam or non-spam, and 57 integer-valued variables. All the features are log-transformed by $\log(1+x)$ before applying the methods.  We split the data into a training set consisting of 3065 observations and a test set consisting of 1536 observations. The methods are tuned on the training set using 10-fold cross validation before predicting on the test set. The results are shown in Figure \ref{fig:comparison_spam}.
    \begin{figure}
      \begin{center}
        \includegraphics[height=3in, width=3in]{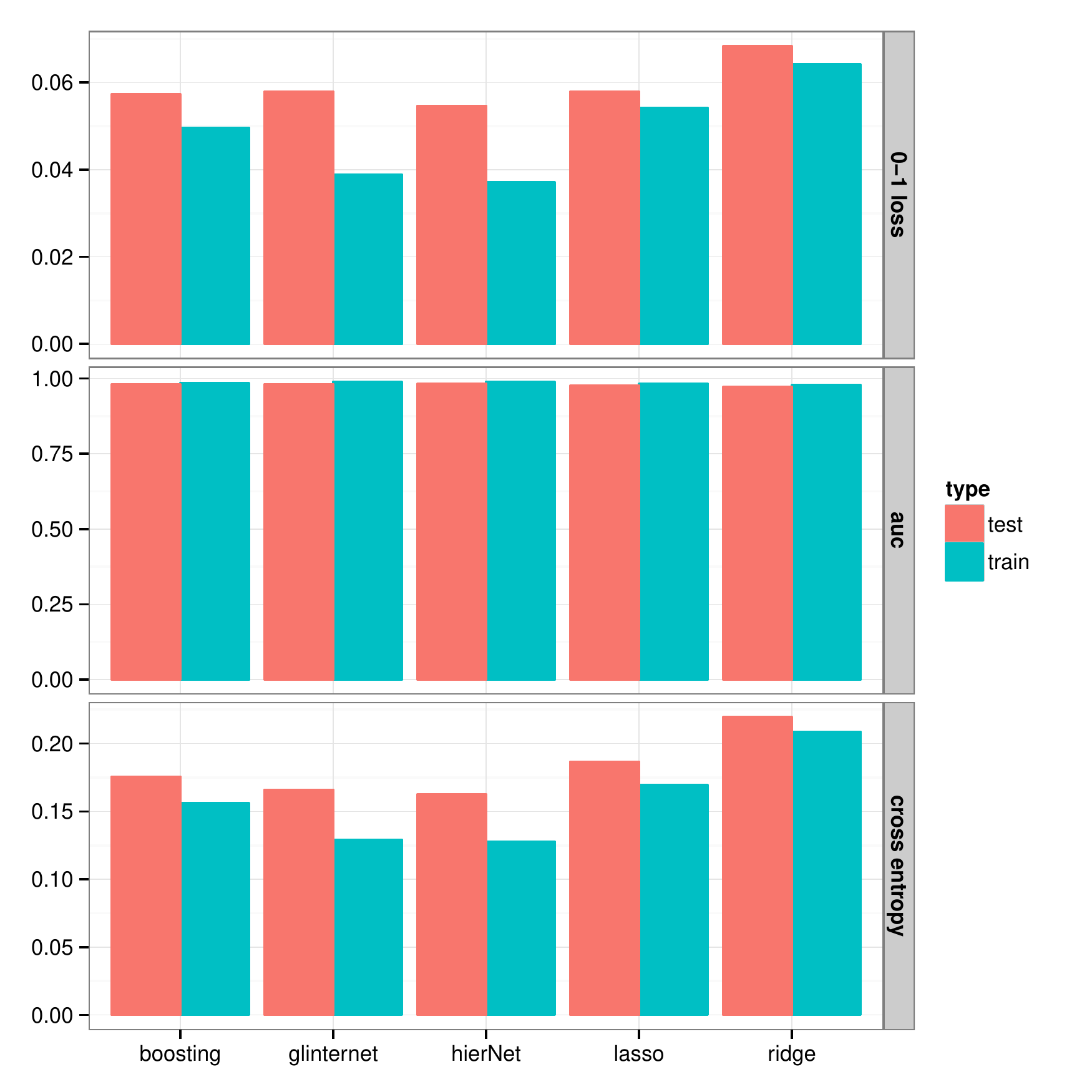}
      \end{center}
      \caption{Performance on the Spambase data.}
      \label{fig:comparison_spam}
    \end{figure}
  \end{subsection}
  \begin{subsection}{Dorothea}
    Dorothea is one of the 5 datasets from the NIPS 2003 Feature Learning Challenge, where the goal is to predict if a chemical molecule will bind to a receptor target. There are 100000 binary features that describe three-dimensional properties of the molecules, half of which are probes that have nothing to do with the response. The training, validation, and test sets consist of 800, 350, and 800 observations respectively. More details about how the data were prepared can be found at http://archive.ics.uci.edu/ml/datasets/Dorothea.

    We run \textsc{glinternet} with screening on 1000 main effects, which results in about 100 million candidate interaction pairs. The validation set was used to tune all the methods. We then predict on the test set with the chosen models and submitted the results online for scoring. The best model chosen by \textsc{glinternet} made use of 93 features, compared with the 9 features chosen by $L_1$-penalized logistic regression (lasso). Figure \ref{fig:comparison_dorothea} summarizes the performance for each method. We see that \textsc{glinternet} has a slight advantage over the lasso, indicating that interactions might be important for this problem. Boosting did not perform well in our false discovery rate simulations, which could be one of the reasons why it does not do well here despite taking interactions into account.\\
    \begin{figure}
      \begin{center}
        \includegraphics[height=3in, width=3in]{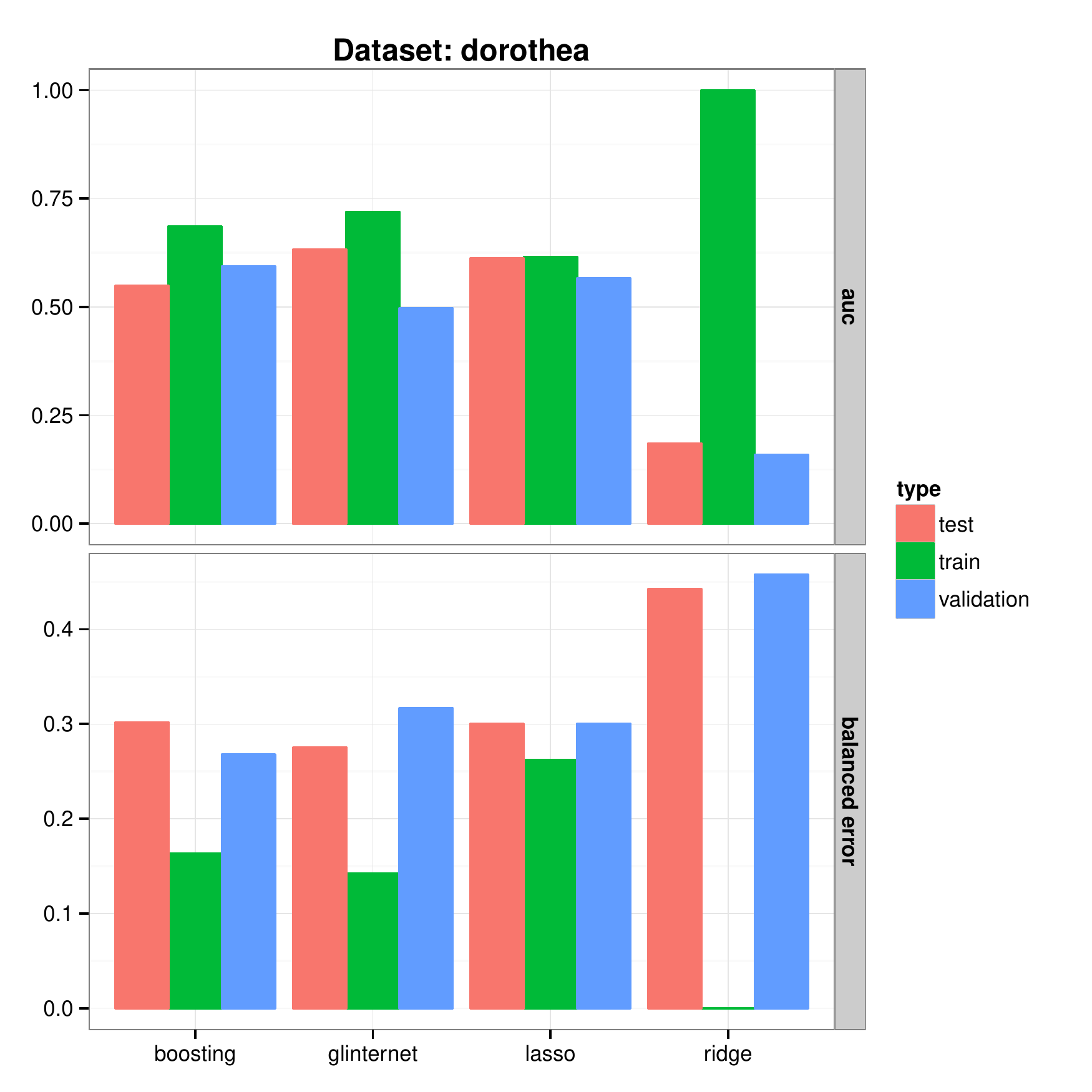}
      \end{center}
      \caption{Performance on dorothea}
      \label{fig:comparison_dorothea}
    \end{figure}
  \end{subsection}
  \begin{subsection}{Genome-wide association study}
    We use the simulated rheumatoid arthritis data (replicate 1) from Problem 3 in Genetic Analysis Workshop 15. Affliction status was determined by a genetic/environmental model to mimic the familial pattern of arthritis; full details can be found in \cite{Miller:2007:BMC}. The authors simulated a large population of nuclear families consisting of two parents and two offspring. We are then provided with 1500 randomly chosen families with an affected sibling pair (ASP), and 2,000 unaffected families as a control group. For the control families, we only have data from one randomly chosen sibling. Therefore we also sample one sibling from each of the 1,500 ASPs to obtain 1,500 cases.

    There are 9,187 single nucleotide polymorphism (SNP) markers on chromosomes 1 through 22 that are designed to mimic a 10K SNP chip set, and a dense set of 17,820 SNPs on chromosome 6 that approximate the density of a 300K SNP set. Since 210 of the SNPs on chromosome 6 are found in both the dense and non-dense sets, we made sure to include them only once in our analysis. This gives us a total of $9187-210+17820=26797$ SNPs, all of which are 3-level categorical variables.

    We are also provided with phenotype data, and we include sex, age, smoking history, and the DR alleles from father and mother in our analysis. Sex and smoking history are 2-level categorical variables, while age is continuous. Each DR allele is a 3-level categorical variable, and we combine the father and mother alleles in an unordered way to obtain a 6-level DR variable. In total, we have 26,801 variables and 3,500 training examples.

    We run \textsc{glinternet} (without screening) on a grid of values for $\lambda$ that starts with the empty model. The first two variables found are main effects:
    \begin{itemize}
    \item SNP6\_305
    \item denseSNP6\_6873
    \end{itemize}
    Following that, an interaction denseSNP6\_6881:denseSNP6\_6882 gets picked up. We now proceed to analyze this result.

    Two of the interactions listed in the answer sheet provided with the data are: locus A with DR, and locus C with DR. There is also a main effect from DR. The closest SNP to the given position of locus A is SNP16\_31 on chromosome 16, so we take this SNP to represent locus A. The given position for locus C corresponds to denseSNP6\_3437, and we use this SNP for locus C. While it looks like none of the true variables are to be found in the list above, these discovered variables have very strong association with the true variables.

    If we fit a linear logistic regression model with our first discovered pair denseSNP6\_6881:denseSNP6\_6882, the main effect denseSNP6\_6882 and the interaction terms are both significant:\\
    \begin{table}[H]
      \centering
      \begin{tabular}{rrrrr}
        \hline
        & Df & Dev & Resid. Dev & $\Pr(>\chi^2)$\\
        \hline
        NULL & & &4780.4&\\
        denseSNP6\_6881&1&1.61&4778.7&0.20386\\
        denseSNP6\_6882&2&1255.68&3523.1&$<$2e-16\\
        denseSNP6\_6881:denseSNP6\_6882&1&5.02&3518.0&0.02508\\
        \hline
      \end{tabular}
      \caption{Anova for linear model fitted to first interaction term that was discovered.}
      \label{tab:table1}
    \end{table}
    A $\chi^2$ test for independence between denseSNP6\_6882 and DR gives a p-value of less than 1e-15, so that \textsc{glinternet} has effectively selected an interaction with DR. However, denseSNP6\_6881 has little association with loci A and C. The question then arises as to why we did not find the true interactions with DR. To investigate, we fit a linear logistic regression model separately to each of the two true interaction pairs. In both cases, the main effect DR is significant (p-value $<1e-15$), but the interaction term is not:\\
    \begin{table}[H]
      \centering
      \begin{tabular}{rrrrr}
        \hline
        & Df & Dev & Resid. Dev & $\Pr(>\chi^2)$\\
        \hline
        NULL & & &4780.4&\\
        SNP16\_31&2&3.08&4777.3&0.2147\\
        DR&5&2383.39&2393.9&$<$2e-16\\
        SNP16\_31:DR&10&9.56&2384.3&0.4797\\
        \hline
        \hline
        NULL & & &4780.4&\\
        denseSNP6\_3437&2&1.30&4779.1&0.5223\\
        DR&5&2384.18&2394.9&$<$2e-16\\
        denseSNP6\_3437:DR&8&5.88&2389.0&0.6604\\
        \hline
      \end{tabular}
      \caption{Anova for linear logistic regression done separately on each of the two true interaction terms.}
      \label{tab:anova_two_true_interactions}
    \end{table}
    Therefore it is somewhat unsurprising that \textsc{glinternet} did not pick these interactions.

    This example also illustrates how the the group-lasso penalty in \textsc{glinternet} helps in discovering interactions (see Section \ref{sec:properties_of_the_glinternet_estimators}). We mentioned above that denseSNP6\_6881:denseSNP6\_6882 is significant if fit by itself in a linear logistic model (Table \ref{tab:table1}). But if we now fit this interaction \textit{in the presence} of the two main effects SNP6\_305 and denseSNP6\_6873, it is \textit{not} significant:\\
    \begin{table}[H]
      \centering
      \begin{tabular}{rrrrr}
        \hline
        & Df & Dev  & Resid. Dev & $\Pr(>\chi^2)$\\
        \hline
        NULL & & &4780.4&\\
        SNP6\_305&2&2140.18&2640.2&$<$2e-16\\
        denseSNP6\_6873&2&382.61&2257.6&$<$2e-16\\
        denseSNP6\_6881:denseSNP6\_6882&4&3.06&2254.5&0.5473\\
        \hline
      \end{tabular}
    \end{table}
    This suggests that fitting the two main effects fully has explained away most of the effect from the interaction. But because \textsc{glinternet} regularizes the coefficients of these main effects, they are not fully fit, and this allows \textsc{glinternet} to discover the interaction.

    The anova analyses above suggest that the true interactions are difficult to find in this GWAS dataset. Despite having to search through a space of about 360 million interaction pairs, \textsc{glinternet} was able to find variables that are strongly associated with the truth. This illustrates the difficulty of the interaction-learning problem: even if the computational challenges are met, the statistical issues are perhaps the dominant factor.
  \end{subsection}
\end{section}

\begin{section}{Algorithm details}\label{sec:algorithm_details}
  We describe the algorithm used in \textsc{glinternet} for solving the group-lasso optimization problem. Since the algorithm applies to the group-lasso in general and not specifically for learning interactions, we will use $\mathbf{Y}$ as before to denote the $n$-vector of observed responses, but $\mathbf{X}=[\mathbf{X}_1\quad \mathbf{X}_2\quad\ldots\quad\mathbf{X}_p]$ will now denote a generic feature matrix whose columns fall into $p$ groups.
  \begin{subsection}{Defining the group penalties $\gamma$}\label{sec:defining_the_group_penalties}
    Recall that the group-lasso solves the optimization problem
    \begin{eqnarray}
      \argmin_\beta \mathcal{L}(\mathbf{Y},\mathbf{X};\beta) + \lambda\sum_{i=1}^p\gamma_i\|\beta_i\|_2,
    \end{eqnarray}
    where $\mathcal{L}(\mathbf{Y},\mathbf{X};\beta)$ is the negative log-likelihood function. This is given by
    \begin{eqnarray}
      \mathcal{L}(\mathbf{Y},\mathbf{X};\beta) = \frac{1}{2n}\left\|\mathbf{Y}-\mathbf{X}\beta\right\|_2^2
    \end{eqnarray}
    for squared error loss, and
    \begin{eqnarray}
      \mathcal{L}(\mathbf{Y},\mathbf{X},\beta) = -\frac{1}{n}\left[\mathbf{Y}^T(\mathbf{X}\beta) - \ind^T\log(\ind+\exp(\mathbf{X}\beta))\right]
    \end{eqnarray}
    for logistic loss (log and exp are taken component-wise). Each $\beta_i$ is a vector of coefficients for group $i$. When each group consists of only one variable, this reduces to the lasso.

    The $\gamma_i$ allow us to penalize some groups more (or less) than others. We want to choose the $\gamma_i$ so that if the signal were pure noise, then all the groups are equally likely to be nonzero. Because the quantity $\|\mathbf{X}_i^T(\mathbf{Y}-\hat{\mathbf{Y}})\|_2$ determines whether the group $\mathbf{X}_i$ is zero or not (see the KKT conditions (\ref{eq:kkt_conditions})), we define $\gamma_i$ via a null model as follows. Let $\epsilon \sim (0, I)$. Then we have
    \begin{eqnarray}
      \gamma_i^2 & = & \E\|\mathbf{X}_i^T\epsilon\|_2^2\\
      & = & \tr \mathbf{X}_i^T\mathbf{X}_i\\
      & = & \|\mathbf{X}_i\|_F^2.
    \end{eqnarray}
    Therefore we take $\gamma_i=\|\mathbf{X}_i\|_F$, the frobenius norm of the matrix $\mathbf{X}_i$. In the case where the $\mathbf{X}_i$ are orthonormal matrices with $p_i$ columns, we recover $\gamma_i=\sqrt{p_i}$, which is the value proposed in \cite{Yuan:2006:JRSS}. In our case, the indicator matrices for categorical variables all have frobenius norm equal to $\sqrt{n}$, so we can simply take $\gamma_i=1$ for all $i$. The case where continuous variables are present is not as straightforward, but we can normalize all the groups to have frobenius norm one, which then allows us to take $\gamma_i=1$ for $i=1,\ldots,p$.
  \end{subsection}

  \begin{subsection}{Fitting the group-lasso}\label{sec:fitting_the_group_lasso}
    Fast iterative soft thresholding (FISTA) \cite{Beck:2009:FIS} is a popular approach for computing the lasso estimates. This is essentially a first order method with Nesterov style acceleration through the use of a momentum factor. Because the group-lasso can be viewed as a more general version of the lasso, it is unsurprising that FISTA can be adapted for the group-lasso with minimal changes. This gives us important advantages:
    \begin{enumerate}
    \item FISTA is a generalized gradient method, so that there is no Hessian involved
    \item virtually no change to the algorithm when going from squared error loss to logistic loss
    \item gradient computation and parameter updates can be parallelized
    \item can take advantage of adaptive momentum restart heuristics.
    \end{enumerate}
    Adaptive momentum restart was introduced in \cite{Candes:2012:FCM} as a scheme to counter the ``rippling'' behaviour often observed with accelerated gradient methods. They demonstrated that adaptively restarting the momentum factor based on a gradient condition can dramatically speed up the convergence rate of FISTA. The intuition is that we should reset the momentum to zero whenever the gradient at the current step and the momentum point in different directions. Because the restart condition only requires a vector multiplication with the gradient (which has already been computed), the added computational cost is negligible. The FISTA algorithm with adaptive restart is given below.
    
    \begin{algorithm}
      \SetKwInOut{Input}{input}\SetKwInOut{Output}{output}
      \caption{FISTA with adaptive restart}
      \label{algorithm:fista}
      \Input{Initialized parameters $\beta^{(0)}$, feature matrix $\mathbf{X}$, observations $\mathbf{Y}$, regularization parameter $\lambda$, step size $s$.}
      \Output{$\hat{\beta}$}
      \BlankLine
      Initialize $x^{(0)}=\beta^{(0)}$ and $\rho_0=1$.
      \BlankLine
      \For{$k=0,1,\ldots,$}{
        $g^{(k)} = -\mathbf{X}^T(\mathbf{Y}-\mathbf{X}\beta^{(k)})$\;
        $x^{(k+1)} = \left(\ind-\frac{s\lambda}{\|\beta^{(k)}-sg^{(k)}\|_2}\right)_+\left(\beta^{(k)}-sg^{(k)}\right)$\;
        $\rho_k = (\beta^{(k)}-x^{(k+1)})^T(x^{(k+1)}-x^{(k)}) > 0$ ? 1 : $\rho_k$\;
        $\rho_{k+1} = (1+\sqrt{1+4\rho_k^2})/2$\;
        $\beta^{(k+1)} = x^{(k+1)} + \frac{\rho_k-1}{\rho_{k+1}}(x^{(k+1)}-x^{(k)})$\;
      }
    \end{algorithm} 
    At each iteration, we take a step in the direction of the gradient with step size $s$. We can get an idea of what $s$ should be by looking at the majorized objective function about a fixed point $\beta_0$:
    \begin{eqnarray}
      \label{eq:majorized_objective}
      M(\beta) = \mathcal{L}(\mathbf{Y},\mathbf{X};\beta_0) + (\beta-\beta_0)^Tg(\beta_0) + \frac{1}{2s}\|\beta-\beta_0\|_2^2 + \lambda\sum_{i=1}^p\|\beta_i\|_2.
    \end{eqnarray}
    Here, $g(\beta_0)$ is the gradient of the negative log-likelihood $\mathcal{L}(\mathbf{Y},\mathbf{X};\beta)$ evaluated at $\beta_0$. Majorization-minimization schemes for convex optimization choose $s$ sufficiently small so that the LHS of (\ref{eq:majorized_objective}) is upper bounded by the RHS. One strategy is to start with a large step size, then backtrack until this condition is satisfied. We use an approach that is mentioned in \cite{Candes:2011:MPC} that adaptively initializes the step size with
    \begin{eqnarray}
      \label{eq:step_size_initialization}
      s = \frac{\|\beta^{(k)}-\beta^{(k-1)}\|_2}{\|g_k-g_{k-1}\|_2}.
    \end{eqnarray}
    We then backtrack from this initialized value if necessary by multiplying $s$ with some $0<\alpha<1$. The interested reader may refer to \cite{Simon:2013:JCGS} for more details about majorization-minimization schemes for the group-lasso.
  \end{subsection}
\end{section}

\begin{section}{Discussion}\label{sec:discussion}
  We introduced \textsc{glinternet}, a method for learning linear interaction models that satisfy strong hierarchy. We demonstrated that the method is comparable with past approaches, but has the added advantage of being able to accomodate both categorical and continuous variables on a larger scale. We illustrated the method with several examples using real and simulated data, and also showed that \textsc{glinternet} can be applied to genome wide association studies.

  \textsc{glinternet} is available on CRAN as a package for the statistical software R.  
\end{section}

\newpage
\bibliographystyle{apalike}
\bibliography{bib}

\end{document}